\documentclass[journal,twoside,web]{ieeecolor}



\usepackage{amsthm}
\usepackage{mathrsfs}
\usepackage{multirow} 
\usepackage{generic}
\usepackage{amssymb}
\usepackage{pifont}
\newcommand{\xmark}{\ding{55}}%

\usepackage{amsmath,amssymb,amsfonts}
\usepackage{graphicx}
\usepackage{lipsum}
\usepackage{algorithm,algorithmic}
\usepackage{booktabs}

\usepackage{textcomp}

\usepackage{setspace}
\usepackage{mathtools}
\usepackage{enumerate}

\usepackage{caption}
\usepackage{subcaption}
\usepackage{algorithm}
\usepackage{float}
\newtheorem{definition}{Definition}
\newtheorem{theorem}{Theorem}

\newtheorem{lemma}{Lemma}
\newtheorem{corollary}{Corollary}

\theoremstyle{definition}
\newtheorem{remark}{Remark}
\newtheorem{proposition}{Proposition}
\usepackage[colorlinks=true, linkcolor=blue, citecolor=blue, urlcolor=blue]{hyperref}

\makeatletter
\let\NAT@parse\undefined
\makeatother
\usepackage{cite}

\newtheorem{innerexample}{Example}[section]

\newcommand{\abs}[1]{\left\lvert#1\right\rvert}

\newcommand{\R}{\mathbb{R}}

\newcommand{\ddt}[1]{\frac{\mathrm{d}#1}{\mathrm{d}t}}

    {\begin{innerexample}[#1]\pushQED{\qed}}%
    {\popQED\end{innerexample}}

\renewenvironment{proof}[1][Proof]{\par\pushQED{\qed}\normalfont\textit{#1.}~}{\popQED\par}

\def\BibTeX{{\rm B\kern-.05em{\sc i\kern-.025em b}\kern-.08em
T\kern-.1667em\lower.7ex\hbox{E}\kern-.125emX}}
\begin{document}

\title{Nonholonomic Robot Parking by Feedback--- Part I: Modular Strict CLF Designs}

\author{Velimir Todorovski$^{1}$, 
        Kwang Hak Kim$^{1}$, Alessandro Astolfi$^{2}$, 
        and Miroslav Krsti\'{c}$^{1}$
\thanks{This work was supported in part by the Office of Naval Research under Grant No. N00014-23-1-2376, in part by the Air Force Office of Scientific Research under Grant No. FA9550-23-1-0535, and in
part by the National Science Foundation under Grant No. ECCS-2151525. The results and opinions in this paper are solely of the authors and do not reflect the position or the policy of the U.S. Government or the National Science
Foundation.}%
\thanks{$^{1}$ V. Todorovski, K. Kim, and M. Krsti\'{c} are with the Department of Mechanical and Aerospace Engineering, UC San Diego, 9500 Gilman Drive, La Jolla, CA, 92093-0411, {\tt\small \{vtodorovski,kwk001,krstic\}@ucsd.edu}.}%
\thanks{$^{2}$ A. Astolfi is with  Imperial College London, London,
United Kingdom and the University of Rome Tor Vergata, Rome, Italy,
{\tt\small a.astolfi@imperial.ac.uk}}%
}

\maketitle

\begin{abstract}
It has been known in the robotics literature since about 1995 that, in polar coordinates, the nonholonomic unicycle is asymptotically stabilizable by smooth feedback, even globally. 
We introduce a modular design framework that selects the forward velocity to decouple the radial coordinate, allowing the steering subsystem to be stabilized independently. Within this structure, we develop families of feedback laws using passivity, backstepping, and integrator forwarding. Each law is accompanied by a strict control Lyapunov function, including barrier variants that enforce angular constraints. These strict CLFs provide constructive class KL convergence estimates and enable eigenvalue assignment at the target equilibrium. The framework generalizes and extends prior modular and nonmodular approaches, while preparing the ground for inverse optimal and adaptive redesigns in the sequel paper.

\end{abstract}


\allowdisplaybreaks




\section{Introduction}
Unicycle---the canonical nonholonomic system---and in particular its parking task, serves as a benchmark for nonlinear stabilization.
While controllable, Brockett~\cite{brockett1983asymptotic} establishes that the unicycle is not stabilizable by continuous static feedback, and Ryan \cite{ryan1994brockett}, as well as Coron and Rosier~\cite{coron1994relation}, that stabilization is impossible even if such a feedback is permitted to be discontinuous. 

\subsection{\em Resuts preview}  Studying stabilization in polar coordinates, with control by steering and bidirectional velocity, 
we introduce a modular methodology for the design of globally asymptotically stabilizing feedback for the unicycle. For four different state spaces, we produce  three families of design methods---passivity-based, backstepping, and one applied here for the first time for the unicycle: integrator forwarding. The design module for stabilization of the two polar angles by steering is independent from the design module for stabilization of the distance-to-target by velocity. 
The control Lyapunov functions (CLFs) we design are strict, global on their (non-Euclidean) state spaces, and are in fact barrier CLFs, preventing angle windup and state constraint violation.  


\subsection{Literature}

\paragraph{Designs in polar coordinates} While of the most recent vintage, designs that circumvent Brockett's condition have resulted in arguably the simplest to design, analyze, and implement feedback laws. It is in this coordinates frame that we conduct our designs as well. 
Preceded by 
Badreddin and Mansour~\cite{badreddin1993fuzzy}, whose linear design in polar coordinates is rigorously characterized by  Astolfi~\cite{astolfi1999exponential}, a landmark global design with bidirectional forward velocity by Aicardi et al.~\cite{aicardi1995} 
employs a non-strict energy-like Lyapunov function and a non-quantitative Barbalat convergence argument, lacking $\mathcal{KL}$-estimates.
Recently, a Lyapunov augmentation with semiglobal strictness~\cite{wang24_force_controlled_safestable} has been produced for the controller in~\cite{aicardi1995}. 
A noteworthy result in the polar coordinates  by Restrepo et al.~\cite{restrepo2020leader}
achieves global exponential stabilization by backstepping  and with a strict global CLF. The steering and velocity feedback designs are interlaced---nonmodular and complex---and their lack of constraints on angles permits angle winding. 
Furthermore, its forward velocity is unidirectional. Human-like parallel parking maneuvers are typically bidirectional, to account for the vehicle dimensions not being negligible. 

\paragraph{Time-varying feedbacks}
Sontag and Sussmann~\cite{Sontag1980Remarks} first highlighted the role of time-varying feedback in achieving global asymptotic stabilization of one-dimensional systems with drift, showing that such stabilizers can be chosen to be periodic in time.
Since the Brockett–Ryan-Coron–Rosier conditions only hold for autonomous systems, asymptotic stabilization of the unicycle in the Cartesian state space can be achieved only by explicitly incorporating time into the feedback law. Such a control law is introduced for the nonholonomic integrator by Samson~\cite{samson1990velocity}. The seminal work by Coron~\cite{coron1992global_controllable} establishes the existence of smooth periodic globally stabilizing time-varying feedback laws for general driftless systems, with specific controllers of this class given in~\cite{coron1993smooth}. 
With the existence established, a number of notable works—such as Pomet~\cite{pomet1992explicit} and \cite{samson1993time}, \cite{samson1995control}, \cite{d2019small}, \cite{m2002exponential}, \cite{morin1999design} \cite{zuyev2016exponential}—have proposed various time-varying control laws applicable to a broad class of driftless nonholonomic systems including the unicycle. 
Time-dependent controllers may be 
non-robust to delays and timing synchronization and have oscillatory transients and complicated stability analysis. 



\paragraph{Discontinuous feedback (convergence sans stability)} 
Although the unicycle cannot be stabilized, discontinuous feedback can be designed to ensure exponential attractivity to the origin.
Some highlights here include the design by Karafyllis and Krstic~\cite{karafyllis2011nonlinear}, as well as the sliding-mode based feedbacks~\cite{bloch1996stabilization_slidingmode}, \cite{astolfi1998discontinuous}, \cite{khennouf1995construction} designed for the nonholonomic integrator or the chained form, which are equivalent representations of the unicycle (see e.g.~\cite{hespanha1999logic}, \cite{murray_chained_form}).
The absence of stability guarantees renders these designs sensitive to small perturbations. 
Furthermore, for unicycle control, these controllers have a much more restrictive regions of attraction in Cartesian coordinates than in polar coordinates.
Using coordinates transformations, such as the polar transformation, in which the unicycle exhibits a removable singularity, provides an alternative way to bypass the obstruction posed by Brockett. In these transformed coordinates, stability results can be achieved; however, in Cartesian coordinates the designed controllers are still discontinuous, and the stability results correspond only to attractivity.  De Wit and S{\o}rdalen~\cite{de1991exponential} were the first to use such a coordinate transformation that maps from the Cartesian space to a coordinate on a circle centered on the vertical axis that passes through the origin and the position of the unicycle and a coordinate that is equivalent to the tangent to this circle. While the proposed feedback law achieves exponential convergence to the origin, stability is absent even in the transformed coordinates. 
Along this line of work, Astolfi~\cite{astolfi1995exponential,astolfi1996discontinuous} proposed more general singular coordinates transformations referred to as $\sigma$-processes and designs a controller for local exponential stabilization of the unicycle in the transformed coordinates, with a region of attraction that excludes starting on the $x$-axis. The result from Astolfi~\cite{astolfi1996discontinuous} on general chained forms can also be applied to the unicycle, but the feedback one gets, although global in the transformed coordinates, incorporates a division by the heading angle, excluding  initial conditions where the vehicle points in the same direction as the target.

\paragraph{Hybrid feedback}
Since Brockett's obstruction does not hold for hybrid systems, the unicycle parking problem has motivated the development of several hybrid feedback strategies~\cite{pomet1992hybrid}. Hespanha et al.~\cite{hespanha1999_hybrid_stabilization,hespanha1999logic} proposed a hysteresis switching algorithm ensuring global exponential stability of the nonholonomic integrator under model uncertainties. However, the induced control jumps and resulting zig-zagging trajectories render such methods impractical, and their state-dependent logic is sensitive to measurement errors.
Prieur et al.~\cite{prieur2003robust,prieur2005robust} developed hybrid schemes combining local and global controllers for chained systems and the nonholonomic integrator, robust to small measurement errors and disturbances but not ISS, and requiring excessive actuation. More recently, Ballaben et al.~\cite{ballaben2025orchestrating} proposed polar-coordinate hybrid controllers for unicycles with camera sensors, proving global asymptotic stability via LaSalle-like arguments.

\paragraph{Nonholonomic tracking}
For nonholonomic systems, stabilization and tracking problem differ fundamentally  \cite{jiang2010controlling}. Tracking employs a persistently exciting reference trajectory, whose time-varying signal circumvents Brockett's obstruction.
Notable works on tracking control for the unicycle include Jiang et al. time-varying state feedback via backstepping~\cite{jiangdagger1997tracking,jiang1999recursive, jiang2001saturated}, showing local and global exponential stability with a nonstrict Lyapunov function, as well as time-varying controllers for parking and tracking that respect input constraints through suitable choice of gains.

\begin{table}[t]
\small
\renewcommand\arraystretch{1.2}
\begin{center}
\begin{tabular}{|l|c|c|l|} 
\hline
\textbf{Design Method} & $\delta\in\mathbb{R}$ & $|\delta| < \pi$ &  \\
\hline
\multirow{4}{*}{\textbf{PASSIVITY}} 
& Genova  & BoPa   & \multirow{2}{*}{ $\gamma \in  \R$}  \\
& Thm.~\ref{thm:unicycle_CLF_polynomial} & Thm.~\ref{thm:unicycle_CLF_BoPA} &
 \\
  \cline{2-4}
& BoLSA  & BAgAl   & \multirow{2}{*}{ $|\gamma| < \pi$}  \\
&Thm.~\ref{thm:unicycle_CLF_BoLSA} &Thm.~\ref{thm:unicycle_CLF_Bagal}  & \\
\hline
\multirow{4}{*}{\textbf{FORWARDING}} 
& GloFo  & \multirow{2}{*}{\xmark} & \multirow{2}{*}{$\gamma \in \mathbb{R}$}  \\ 
& Thm.~\ref{thm:CLF_GloFo} &  & \\
\cline{2-4}
& BoFo  & \multirow{2}{*}{\xmark} & \multirow{2}{*}{$|\gamma| < \pi$} \\ 
& Thm.~\ref{thm:CLF_BoFo} & & \\
\hline
\multirow{4}{*}{\textbf{BACKSTEPPING}} 
& GloBa  & BAR-Fli & \multirow{2}{*}{$\gamma \in \mathbb{R}$} \\
& Thm.~\ref{thm:CLF_Globa} & Thm.~\ref{thm:CLF_BARFLi} &
\\
\cline{2-4}
& \centering \xmark &  \xmark & $|\gamma| < \pi$ \\ 
\hline
\end{tabular}
\end{center}
\caption{Controller acronyms (defined in Sec.~\ref{sec:feedback_summary}) for three design methods applied to the unicycle system. 
The angles $\delta$ and $\gamma$ represent the polar and line-of-sight angles, respectively. }
\label{tab:controllers_summary}
\end{table}

\subsection{Contributions and Organization}

\paragraph{Modular approach} We present a modular approach to solving the unicycle parking problem in polar coordinates. This is achieved by choosing the forward velocity to allow both forward and backward motion, suitable for parking maneuvers. The state representing the Euclidean distance to the target is then decoupled, allowing the remaining dynamics to be controlled independently via the steering. The particular structure of this subsystem allows designing a multitude of steering control laws through passivity, backstepping, and integrator forwarding techniques (see Table~\ref{tab:controllers_summary} for a preview) applied to state-spaces with constraints ranging from unconstrained to those that prevent angular wind-up.

\paragraph{Strict CLFs} The steering laws are accompanied by families of global strict CLFs and barrier variants, ensuring global asymptotic stability on the respective state spaces with quantitative $\mathcal{KL}$-estimates of the convergence rates, while also enabling us to assign the eigenvalues at the unicycle’s target. While asymptotic stabilization is achieved in polar coordinates via smooth feedback, in Cartesian coordinates only attractivity holds, due to the discontinuity of the polar transformation. Through the $\mathcal{KL}$-estimates in polar coordinates, we provide an example of quantitative characterization of the convergence properties of the Cartesian states. 
Another advantage of the modular framework is that it enables the construction of composite Lyapunov functions, which greatly expand the range of possible CLF designs. We leverage this in the second part of the paper to develop closed form optimal controllers in the inverse sense based on the strict CLFs introduced here, as well as nonlinear adaptive controllers that account for wheel slippage. Portions of the results presented appear in \cite{todorovski2025_CLF} and \cite{Krstic2025_fwd}.

In Part II of this paper \cite{Part2_kim2025} we make advances upon~\cite{restrepo2020leader}. While the nonmodular backstepping does not allow the construction of barrier Lyapunov functions that constrain {\em simultaneously both} angular states,  we design barrier CLFs that constrain the polar angle alone and, additionally, achieve global exponential stabilization by bidirectional  velocity actuation. 

\paragraph{Organization} 
After the introduction of the polar coordinate transformation in Sec.~\ref{sec:unicycle_polar_and_cartesian}, we introduce our modular framework in Sec.~\ref{sec-preliminaries}, followed by a summary of our feedback and CLFs designs in Sec.~\ref{sec:feedback_summary}. The construction of composite CLFs is introduced in Sec.~\ref{sec:composite} and the main results of this paper are presented in Sections~\ref{sec-Genova}--\ref{sec:sims}.

\section{Unicycle in Cartesian and Polar Representations}
\label{sec:unicycle_polar_and_cartesian}



\subsection{Unicycle model}

\begin{figure}[t]
\centering
\includegraphics[width=0.65\linewidth]{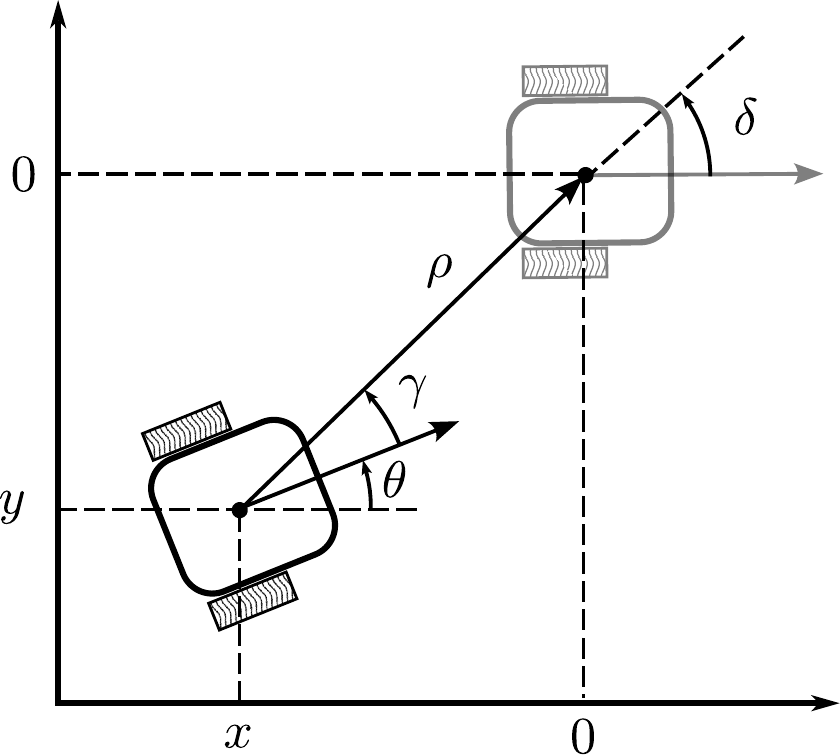}
\caption{Unicycle orientation $(x, y, \theta)$ relative to the goal state $(0, 0, 0)$, and the corresponding polar coordinates $(\rho, \delta, \gamma)$.}
\label{fig:unicycle_cord}
\end{figure}

We consider the unicycle model 
\begin{subequations}
\label{eq:unicycle_cartesian}
\begin{align}
\dot{x} &=  v \cos(\theta)\\
\dot{y} &= v \sin(\theta)  \\
\dot{\theta} &= \omega \,,\label{eq:dot_theta_omega}
\end{align}
\end{subequations}
where $(x,y) \in \R^2$ is the position of the unicycle in Cartesian coordinates, $\theta \in \mathbb{R}$ is the heading angle, $v$ is the forward velocity input, and $\omega$ is the angular velocity input. The unicycle can be represented in polar coordinates with the transformation defined in Table~\ref{tab:polar_coordinates} with the inverse transformation given in Table~\ref{tab:inverse_polar_coordinates}. The resulting polar representation of the unicycle  is given as 
\begin{subequations}\label{eq:unicycle_polar_closed_loop-Gv-1}
    \begin{align}
    \dot{\rho} &= 
    -v \cos\gamma\label{eq:unicycle_polar_rhodot}\,,\\
    \dot{\delta} &= \frac{v}{\rho}  \sin\gamma\label{eq:unicycle_polar_deltadot}\,,\\
    \dot{\gamma} &= \frac{v}{\rho} \sin\gamma -\omega \label{eq:unicycle_polar_gammadot} \,.
    \end{align} 
\end{subequations}


\begin{table*}[t]
\small
\renewcommand\arraystretch{1.2}
\centering
\begin{tabular}{|l|l|l|}
\hline
\textbf{Polar Coordinates}  & \textbf{Description} \\
\hline
$\rho = \sqrt{x^2 + 
y^2}$  & Distance to origin \\
\hline
$\delta = \text{{\rm atan2}}(y ,x )  + \pi= 2\arctan\dfrac{\sqrt{x^2+y^2}-x}{y}+\pi$ & Polar  angle
\\
\hline
$\gamma = \delta - \theta =\text{{\rm atan2}}(y, x) - \theta + \pi$  & Line-of-sight (LoS) angle   \\
\hline
\end{tabular}
\caption{Polar coordinates and their expressions in terms of Cartesian coordinates. The transformation $(x,y,\theta)\mapsto (\rho,\delta,\gamma)$ is discontinuous on $x<0,y=0$ and not defined at $x=y=0$. If the target is at $(x^*, y^*, \theta^*)\neq 0$, the transformation generalizes to $\rho=\sqrt{(x-x^*)^2+(y-y^*)^2}, \delta= \text{{\rm atan2}}(y-y^* ,x-x^* ) -\theta^* + \pi, \gamma= \delta-\theta+\theta^*$.}
\label{tab:polar_coordinates}
\end{table*}

\begin{table*}[t]
\small
\centering
\begin{tabular}{|l|l|l|}
\hline
\textbf{Cartesian Coordinates}  & \textbf{Description} \\
\hline
$x  = -\rho \cos \delta = -\rho(\cos\gamma\cos\theta-\sin\gamma\sin\theta) $  & Horizontal distance to target \\
\hline
$y  = -\rho \sin \delta = -\rho(\sin\gamma\cos\theta+\cos\gamma\sin\theta)$  & Vertical distance to target   \\
\hline
$ \theta  = \delta - \gamma $ & Heading error \\ 
\hline
\end{tabular}
\caption{Inverse transformation from polar to cartesian coordinates. 
The transformation $(\rho,\delta,\gamma)\mapsto  (x,y,\theta)$ is continuous and, when $(\rho,\delta,\gamma)=0$, it holds that $(x,y,\theta)=0$.}
\label{tab:inverse_polar_coordinates}
\end{table*}


\paragraph{Transformation to ``nonholonomic integrator''} The unicycle is transformable into 
$\dot \xi = v-\eta \omega$, 
$\dot \eta = \xi \omega$, 
$\dot\theta = \omega$,
using the 
transformations
$\xi = - \rho\cos\gamma = x \cos\theta +y\sin\theta$ 
and 
$\eta = \rho \sin\gamma= x \sin\theta -y\cos\theta 
$, 
whose inverse is $\rho=\sqrt{\xi^2+\eta^2}, \gamma = - \arctan(\eta/\xi)$, and for which $\cos\gamma= -\xi/\rho, \sin\gamma = \eta/\rho, \delta = \theta-\arctan(\eta/\xi), \tan(\gamma/2) = \eta/(\rho-\xi)$. 
Those of our control designs that are globally stabilizing (Genova, GloFo, GloBa, and their $L_gV$ counterparts) also provide stabilizing solutions for the nonholonomic integrator.

\subsection{Alternative angle definitions}

\paragraph{Stabilization ``in reverse''}
The polar angles in Table~\ref{tab:polar_coordinates} imply that $\delta=0$ places the unicycle behind the target, and $\gamma=0$ aligns it toward the target (Fig.~\ref{fig:unicycle_cord}). Hence, stabilizing controllers naturally guide the unicycle to approach the origin forward from behind the target and ``forward-park''. However, in some cases, it may be preferable for the unicycle to ``reverse-park'' instead.

Consider the alternative definition $\hat{\delta} := {\rm atan2}(y, x)$, $\hat{\gamma} := {\rm atan2}(y, x)$
noting that $\hat{\delta} = 0$ corresponds to being directly in front of the target, while $\hat{\gamma} = 0$ indicates that the unicycle is facing directly away from it. 
This definition is simply a phase shift of $-\pi$ from the original definition in Table~\ref{tab:polar_coordinates} and effectively modifies \eqref{eq:unicycle_polar_closed_loop-Gv-1} to
\begin{align}
\label{eq:unicycle_polar_closed_loop-Gv-1-hats}
\dot{\rho} = 
v \cos\hat\gamma\,,\quad
\dot{\delta} =- v  \frac{\sin\hat\gamma}{\rho} 
\,,\quad
\dot{\gamma} = -v \frac{\sin\hat\gamma}{\rho} -\omega  \,.
\end{align} 
Observe that if we define $\hat{v} = -v$, the system becomes identical to 
\eqref{eq:unicycle_polar_closed_loop-Gv-1}.
Thus, any stabilizing controls $(v,\omega)$ for 
\eqref{eq:unicycle_polar_closed_loop-Gv-1}
can be applied to 
\eqref{eq:unicycle_polar_closed_loop-Gv-1-hats}
by using $\hat{v} = -v$, for reverse-parking.

\paragraph{On the discontinuity of the transformed variables}


Despite smooth feedback being possible in polar coordinates, the discontinuity in the Cartesian coordinates still remains due to the transformation involving $\text{atan2}(y,x)$.
This can be mitigated by redefining $\delta$ and $\gamma$ as $\tilde{\delta} = \mbox{mod}(\hat{\delta} 
, 2\pi) - \pi$  and  $\tilde{\gamma} = \mbox{mod}(\hat{\gamma}, 2\pi) - \pi$
such that all $\delta, \gamma \in \mathbb{R}$ are wrapped to $ \tilde{\delta}, \tilde{\gamma}\in [-\pi, \pi)$, with discontinuities appearing 
when the vehicle points away from the target position (i.e., $|\gamma| = \pi$) or when the vehicle crosses the line $\{x>0,y=0\}$ (i.e., $|\delta| = \pi$). The inverse transform to $x$ and $y$ in rows 1 and 2 of Table~\ref{tab:inverse_polar_coordinates} remains unchanged, but it is important to note that the inverse transform for $\theta$ in row 3 of Table~\ref{tab:inverse_polar_coordinates} only holds in the sense of modulo $2\pi$, i.e., $\theta \;(\mbox{mod} \; 2\pi) = \tilde{\delta} - \tilde{\gamma}$. As will be shown in Section~\ref{sec:feedback_summary}, shifting the discontinuity is highly beneficial since many of the proposed feedback laws repel $\delta$ and $\gamma$ away from $\pm\pi$, unlike the original definition where it appears whenever the vehicle crosses ${x<0, y=0}$.


Additionally, regardless of which angle definition is used, for every controller \textit{there exist} initial conditions that result in a jump away from the origin when crossing either the positive half or the negative half of the $x$-axis, but not both. Such a discontinuity of control is not visually detectable in the trajectories $y(x)$, which have continuous ${\rm d}y/{\rm d}x=\tan(\theta)$, where $\theta(t)=\theta(0)+\int_0^t \omega(\tau){
\rm d}\tau$ remains continuous even at jumps of $\omega$, but have a possible discontinuity only in the curvature ${\rm d}y^2/{\rm d}x^2$, which cannot be visually perceived. 

\section{Preliminaries of the Feedback Design}
\label{sec-preliminaries}

\subsection[Design on rho=0: steering to turn unicycle to theta=0]%
{Design on $\rho=0$: steering to turn unicycle to $\theta=0$}

At $\rho=0$, setting $v=0$, the vehicle remains at the origin and the task is to design the feedback $\omega(\theta)$ to turn the vehicle to $\theta=0$. Additionally, the angles $\delta,\gamma$ are undefined, so there is no point of seeking feedback of the form $\omega(\delta,\gamma)$ when $\rho=0$. 
So, we turn our attention to the system \eqref{eq:dot_theta_omega}, i.e., 
$\dot{\theta} = \omega$. 
If the state space for $\theta$ is $\Sigma_\infty = \mathbb{R}$, the feedback 
$\omega = -k_0 \theta$, $k_0 > 0$
and the Lyapunov function $V= \theta^2$ lead to $\dot V = - 2k_0 V$. If, instead, the state space for $\theta$ is $\Sigma=(-\pi,\pi)$, 
the feedback
$\omega = -k_0\sin\theta$
and the Lyapunov function $V= 4 \tan^2 \frac{\theta}{2}$ lead to $\dot V = - 2k_0 V$, whereas
the feedback
$\omega = -k_0\tan\dfrac{\theta}{2}$
leads to $\dot V = - k_0\left(1+\frac{V}{4}\right) V$. 

\subsection{Forward velocity feedback}
\label{sec:forward_velocity_feedback}

For stabilization alone of the system \eqref{eq:unicycle_polar_closed_loop-Gv-1}, it is not evident that one can choose a feedback law that is either significantly better performing or simpler than the feedback
\begin{equation}
\label{eq-basic-v-control}
\fbox{$v =  k_1 \rho \cos(\gamma)$} = -k_1(x\cos\theta+y\sin\theta)\,,  \quad k_1>0\,.
\end{equation}
Note that \eqref{eq-basic-v-control} is never discontinuous away from the origin, which implies that any potential jump in the control input is caused by the steering feedback law.
The feedback \eqref{eq-basic-v-control} yields the system consisting of the closed-loop subsystem 
\begin{align}
\label{eq:unicycle_polar_closed_loop-Gv-2}
\dot{\rho} &=  
- k_1 \rho \cos^2(\gamma)\,,
\end{align} 
for which $\rho(t)$ exponentially converges to zero when $\cos \gamma\neq 0$, 
and, on the set $\rho>0$, in which the cancellation $\rho/\rho=1$ is valid, the subsystem 
\begin{subequations}
\label{eq:unicycle_polar_closed_loop-Gv-3}
\begin{align}
\dot{\delta} &=  \frac{k_1}{2} \sin(2\gamma) \label{eq:unicycle_polar_closed_loop-Gv-3n}\\
\dot{\gamma} &= \frac{k_1}{2} \sin(2\gamma) -\omega \label{eq:unicycle_polar_closed_loop-Gv-3a} \,,
\end{align} 
\end{subequations}
for which a feedback law depending only on $(\delta,\gamma)$, and applied by the steering input $\omega$, needs to be designed. 
By inspection of \eqref{eq:unicycle_polar_closed_loop-Gv-3a}, it is clear that, when the vehicle is not steered, namely, when $\omega=0$, the $\gamma$-subsystem is unstable at the equilibrium $\gamma=0$ and exponentially stable at the equilibrium $\gamma = \pm\pi/2$ (modulo $2\pi$). In other words, the vehicle, under the feedback \eqref{eq-basic-v-control}, moves along a straight line and settles at a position at which the target position is exactly to its left or to its right, at a distance that depends on the initial conditions $(\rho_0,\gamma_0)$ irrespective of the target heading. The consequence of these observations is that the term $\frac{k_1}{2} \sin(2\gamma)$ in \eqref{eq:unicycle_polar_closed_loop-Gv-3a} is destabilizing and any feedback design for $\omega$ has to either cancel $\frac{k_1}{2} \sin(2\gamma)$ or to appropriately dominate it. For simplicity, we cancel this term and introduce a control law
\begin{equation}
\label{eq-omega-general}
\fbox{$\omega = \dfrac{k_1}{2} \sin(2\gamma) +\tilde\omega$}
\end{equation}
where the control $\tilde\omega$ has to be designed for the resulting system
\begin{subequations}
\label{eq:unicycle_polar_closed_loop-Gv-3-pass}
\begin{align}
\label{eq:unicycle_polar_closed_loop-Gv-3n-pass}
\dot{\delta} &=  \frac{k_1}{2} \sin(2\gamma)\\ \label{eq:unicycle_polar_closed_loop-Gv-3a-pass}
\dot{\gamma} &= -\tilde\omega\,.
\end{align} 
\end{subequations}



\subsection{State spaces and stability} 

In what follows we consider four state spaces, that is 
\begin{align}
\mathcal{S}   &:= \left\{ \rho > 0 \right\} \times \mathcal{T},
& \quad \mathcal{T}   &:= \left\{ \delta \in \mathbb{R},\, \gamma \in \mathbb{R} \right\} \label{eq:ss_S}\\[2pt]
\mathcal{S}_1 &:= \left\{ \rho > 0 \right\} \times \mathcal{T}_1,
& \quad \mathcal{T}_1 &:= \left\{ \delta \in \mathbb{R},\, \abs{\gamma} < \pi \right\} \label{eq:ss_S1}\\[2pt]
\mathcal{S}_2 &:= \left\{ \rho > 0 \right\} \times \mathcal{T}_2,
& \quad \mathcal{T}_2 &:= \left\{ \abs{\delta} < \pi,\, \gamma \in \mathbb{R} \right\} \label{eq:ss_S2}\\[2pt]
\mathcal{S}_3 &:= \left\{ \rho > 0 \right\} \times \mathcal{T}_3,
& \quad \mathcal{T}_3 &:= \left\{ \abs{\delta} < \pi,\, \abs{\gamma} < \pi \right\} \label{eq:ss_S3}
\end{align}
and the following functions (on their respective state-space)
\begin{alignat}{3}
&\mathcal{T}:  &\quad  \Delta &= \delta  &\quad  \Gamma &= \gamma \label{eq:T_Delta_Gamma}\\
&\mathcal{T}_1: &\quad \Delta &= \delta  &\quad \Gamma &= 2\tan\frac{\gamma}{2} \label{eq:T1_Delta_Gamma} \\
&\mathcal{T}_2: &\quad \Delta &= 2\tan\tfrac{\delta}{2} &\quad  \Gamma &= \gamma   \label{eq:T2_Delta_Gamma} \\
&\mathcal{T}_3: &\quad \Delta &= 2\tan\tfrac{\delta}{2}  &\quad \Gamma &= 2\tan\tfrac{\gamma}{2}\,.\label{eq:T3_Delta_Gamma}
\end{alignat}
Then, the state-spaces \eqref{eq:ss_S}--\eqref{eq:ss_S3} are equipped with the metrics
\begin{equation}
|(\rho,\delta,\gamma)|_{\hat{\mathcal{S}}} \coloneqq \rho +\abs{ (\delta,\gamma)}_{\hat{\mathcal{T}}} :=  \rho + |\Delta| + |\Gamma| \label{eq:metric_S_hat}   
\end{equation}
where $\hat{\mathcal{S}} \in \{\mathcal{S}, \mathcal{S}_1, \mathcal{S}_2, \mathcal{S}_3\}$ and $\hat{\mathcal{T}} \in \{\mathcal{T}, \mathcal{T}_1, \mathcal{T}_2, \mathcal{T}_3\}$, and $(\Delta$, $\Gamma)$ are defined according to \eqref{eq:T_Delta_Gamma}--\eqref{eq:T3_Delta_Gamma} depending on $\hat{\mathcal{T}}$.

\begin{table}[t]


\small
\renewcommand\arraystretch{2.2}
\begin{tabular}{ |c|c| } 
\hline
\textbf{State-Space} $\hat{\mathcal{S}}$ & \textbf{CLF} $V(\rho,\delta,\gamma)$ \\
\hline
\hline
$\mathcal{S}$ 
& $\begin{aligned}
  \rho^2 +   \tfrac{1}{2}(\delta^2 + \gamma^2 + 2)^2 - 2 + (\delta + \gamma)^2
\end{aligned}$ \\
\hline
$\mathcal{S}_1$ 
& $\begin{aligned}
    \rho^2 + (\delta + \sin \gamma)^2 + 4\tan^2 \frac{\gamma}{2}
\end{aligned}$ \\
\hline
$\mathcal{S}_2$  
& $\begin{aligned}
   \rho^2 +  \delta^2 + \left(\gamma + \tfrac{1}{2}\arctan\!\left(4\tan \tfrac{\delta}{2}\right)\right)^2
\end{aligned}$ \\
\hline

$\mathcal{S}_3$  
& $\begin{aligned}
  \rho^2 +  &\left(4\tan^2 \tfrac{\delta}{2} + 4\tan^2 \tfrac{\gamma}{2} + 1\right)^3 \\
    &- 1 + \left(2\tan \tfrac{\delta}{2} + 2\tan \tfrac{\gamma}{2}\right)^2
\end{aligned}$ \\
\hline
\end{tabular}

    \caption{Representative CLFs for the unicycle~\eqref{eq:unicycle_polar_closed_loop-Gv-1} on each of the state-spaces $\hat{\mathcal{S}} \in \{\mathcal{S}, \mathcal{S}_1, \mathcal{S}_2, \mathcal{S}_3\}$ defined in \eqref{eq:ss_S}--\eqref{eq:ss_S3}.}
    \label{tab:CLF}
\end{table}

\begin{definition}
\label{def-our-GAS}
Consider the system \eqref{eq:unicycle_polar_closed_loop-Gv-2}, \eqref{eq:unicycle_polar_closed_loop-Gv-3-pass} with a feedback law $\tilde\omega(\gamma,\delta)$ that is continuous on a state space $\mathcal{Q}$ with respect to its metric. 
If there exists a class $\mathcal{KL}$ function $\beta$ such that, for all $t\geq 0$, it holds that $|(\rho(t),\delta(t), \gamma(t))|_{\mathcal{Q}}\leq \beta\left(|(\rho_0,\delta_0, \gamma_0)|_{\mathcal{Q}},t\right)$, we say that the 
point $\rho=\delta = \gamma=0$ is {\em globally asymptotically stable on $\mathcal{Q}$} (GAS on $\mathcal{Q}$). 
\end{definition}


\section{Summary of the Stabilizing Feedback Designs}\label{sec:feedback_summary}


We summarize the principal control designs developed in the paper: four passivity-based designs,  two integrator forwarding designs, and two backstepping designs.
%
%
Dependent on the state-space, all feedback laws employ the functions \eqref{eq:T_Delta_Gamma}--\eqref{eq:T3_Delta_Gamma} and representative CLFs are provided in Table~\ref{tab:CLF}.
Throughout this section, the term controllers  \textit{'region of attraction'} refers to the region of attraction of the system \eqref{eq:unicycle_polar_closed_loop-Gv-1} under the control laws \eqref{eq-basic-v-control}, \eqref{eq-omega-general}, and the specific steering controller $\tilde{\omega}$ introduced in this section.
In all cases, the gains $k_i$ are  positive constants. Additionally, we define the function $\mbox{sinc}$ as
\begin{equation}
\mbox{sinc}(a) \coloneqq \frac{\sin a}{a}\,, \quad \text{if} \quad a \ne 0, \quad \text{and} \quad \mbox{sinc}(0) = 1,
\end{equation}
which is bounded and continuous. 

\subsection{Passivity-inspired controllers}
Consider \eqref{eq:unicycle_polar_closed_loop-Gv-3-pass}, as well as the storage functions $\Delta^2$ and $\Gamma^2$ which, on our four state spaces, are
given by \eqref{eq:T_Delta_Gamma}--\eqref{eq:T3_Delta_Gamma}.
Regardless of the state space, the storage functions satisfy the `dissipation' relations 
\begin{eqnarray}
  \frac{{\rm d} \Delta^2}{ {\rm d} t} &=&   (\Delta^2)^{\prime} \frac{k_1\sin(2\gamma)}{2}  \\
  \frac{{\rm d} \Gamma^2}{ {\rm d} t} &=&  -(\Gamma^2)^{\prime}  \tilde{\omega}
\end{eqnarray}
The $\Delta$-system is passive from the input $\frac{\sin(2\gamma)}{2}$ to its output $(\Delta^2)'
$, which through $\tilde{\omega}$ can be made the input to the $\Gamma$-system. 
The goal is to design a feedback $\tilde\omega(\Delta,\Gamma)$ so that the $\Gamma$-system is (strictly) passive from its input $(\Delta^2)
'$ to its output $-\frac{\sin(2\gamma)}{2}$. 
To achieve this, $\tilde\omega$ is chosen as
\begin{eqnarray}
    \tilde{\omega} &=& \frac{1}{(\Gamma^2)'}
\left[ k_2 \Gamma^2 + k_3 \dfrac{\sin(2\gamma)}{2}(\Delta^2)'\right]\,. 
\label{eq:passivity_general_control}
\end{eqnarray}
In closed loop, one has
\begin{eqnarray}
  \frac{{\rm d} \Delta^2}{ {\rm d} t} &=&   (\Delta^2)^{\prime} \frac{k_1 \sin(2\gamma)}{2}, \label{eq:Delta_squared} \\
  \frac{{\rm d} \Gamma^2}{ {\rm d} t} &=&  -k_2\Gamma^2 - k_3(\Delta^2)'\dfrac{\sin(2\gamma)}{2}.\label{eq:Gamma_squared}
\end{eqnarray}
The first relation establishes passivity from $\frac{\sin(2\gamma)}{2}$ to $(\Delta^2)^{\prime} $, and the second relation {\em strict} passivity from $(\Delta^2)^{\prime} $ to $-\frac{\sin(2\gamma)}{2}$.
The sum of the storage functions of the subsystems, namely, the energy function
\begin{equation}
    U = \Delta^2(\delta) + q^2\Gamma^2(\gamma), \quad q = \sqrt{k_1/k_3} \label{eq:U_common_passivity}
\end{equation}
has the time derivative along \eqref{eq:Delta_squared} and \eqref{eq:Gamma_squared} as $\dot U =  - 2k_2q^2\Gamma^2 \le 0$.
The following two designs adopt this approach using the CLF
\begin{equation}
V = \rho^2 +k_3 \left(1+
\frac{2q^2+U}{2qk_2}\right)U
+ (\Delta+q\Gamma 
)^2\,,
\label{eq:V_CLF_polynomial_summary}
\end{equation}
where $U$ is defined in \eqref{eq:U_common_passivity}.
\paragraph{Genova controller ($\mathcal{T}$)} In this paper we first revisit the steering control law introduced in \cite{aicardi1995} which can be derived from the general passive control law  \eqref{eq:passivity_general_control} by setting $\Delta$ and $\Gamma$ as in \eqref{eq:T_Delta_Gamma} which results in the expression 
\begin{equation}
\fbox{$\tilde\omega =  k_2 \gamma + k_3 
\mbox{sinc}(2\gamma)\delta.$} 
\label{eq:angular_velocity_genova}
\end{equation}
Even though this controller is not our design, we design a strict CLF for this controller which is given by substituting \eqref{eq:T_Delta_Gamma} for $\Delta$ and $\Gamma$ in \eqref{eq:V_CLF_polynomial_summary}.

\paragraph{BoLSA controller ($\mathcal{T}_1$)} The next passivity-based controller is obtained from \eqref{eq:passivity_general_control} by substituting $\Delta$ and $\Gamma$ from \eqref{eq:T1_Delta_Gamma}, and it can be expressed as
\begin{equation}
\label{eq-control-bounded-in-gamma}
\fbox{$\tilde\omega =  k_2\sin\gamma +  \displaystyle
\frac{k_3\cos\gamma}{\left(1+\displaystyle\tan^2\frac{\gamma}{2}\right)^2}\delta$.} 
\end{equation}
This is bounded in the LoS angle $\gamma$ and we, consequently, refer to it as the Bounded-in-LoS Angle (BoLSA, pronounced `bolsa') controller. BoLSA's region of attraction is the state space 
$\mathcal{S}_1$ defined in \eqref{eq:ss_S1},
which includes all the initial headings except  exactly those away from the target. In other words, the algorithm achieves stable parking while, in the process, never ``turning its back'' against the target position. 
The strict CLF for this controller is obtained by substituting the expressions for $\Delta$ and $\Gamma$ from \eqref{eq:T1_Delta_Gamma} in \eqref{eq:V_CLF_polynomial_summary}.

Interestingly, the controllers for the state-spaces $\mathcal{T}_2$ and $\mathcal{T}_3$ share the strict CLF 
\begin{equation}
\label{eq:CLF_ABounD_summary}
V = \rho^2 + \frac{a}{3k_2 q^2} \left[ (1+U)^3 -1\right] + 
\left (\Delta + q \Gamma \right)^2,
\end{equation}
where $a = \max\{k_1 q,   \sqrt{k_1 k_2}\}$.
\paragraph{BoPA controller ($\mathcal{T}_2$)}
On the state-space $\mathcal{T}_2$, the passivity-based controller is obtained by substituting \eqref{eq:T2_Delta_Gamma} in \eqref{eq:passivity_general_control} which results in
\begin{equation}
\label{eq:Bopa_controller}
\fbox{$\tilde\omega = k_2\gamma + 2k_3 \mbox{sinc}(2\gamma) \left(1+\tan^2 \dfrac{\delta}{2} \right)\tan \dfrac{\delta}{2}$.}
\end{equation}
This is bounded in the polar angle $\delta$ and we refer to it as Bounded-in-Polar-Angle (BoPA) controller. BoPa's region of attraction is the state space \eqref{eq:ss_S2}
which entails all positions except the nonnegative half of the $x$-axis. In other words, the algorithm achieves stable parking while never crossing exactly in front of the parking target. 
The CLF for this controller is obtained by substituting \eqref{eq:T2_Delta_Gamma} in \eqref{eq:CLF_ABounD_summary}.

\paragraph{BAgAl controller ($\mathcal{T}_3$)} The last passivity-based controller is given by substituting \eqref{eq:T3_Delta_Gamma} in  \eqref{eq:passivity_general_control}, which results in
\begin{equation}
\label{eq-control-bounded-in-gamma-delta}
\fbox{$\tilde\omega = k_2\sin\gamma +  \displaystyle
\frac{ 2k_3\cos\gamma }{\left(1+\displaystyle\tan^2\frac{\gamma}{2}\right)^2} \left(1+\tan^2 \frac{\delta}{2} \right)\tan \frac{\delta}{2}$.}
\end{equation}
This is bounded with respect to the polar angle $\delta$ and the LoS angle $\gamma$, and we refer to it as the Bounding-Angles Algorithm Design (BAgAl, pronounced “bagel”) controller. BAgAl's region of attraction is the state space $\mathcal{S}_3$ defined in \eqref{eq:ss_S3}
which includes all initial headings except those facing exactly opposite to the target, and all positions except those on the nonnegative half of the $x$-axis. The algorithm achieves stable parking while never “turning its back” on the target or crossing directly in front of it.
The CLF for this controller is obtained by substituting \eqref{eq:T3_Delta_Gamma} in \eqref{eq:CLF_ABounD_summary}.

\subsection{Forwarding controllers}
Consider again the system \eqref{eq:unicycle_polar_closed_loop-Gv-3-pass}. Observe that \eqref{eq:unicycle_polar_closed_loop-Gv-3n-pass} depends solely on the LoS angle $\gamma$, and with \eqref{eq:unicycle_polar_closed_loop-Gv-3a-pass} involving only the input $\tilde{\omega}$, the overall system \eqref{eq:unicycle_polar_closed_loop-Gv-3-pass} exhibits a \textit{strict feedforward} structure, enabling the use of the integrator forwarding method~\cite{sepulchre1997forwarding,mazenc2002adding, krstic2004feedback}. The following two controllers utilize a forwarding transformation, with the corresponding CLFs for each case given by
\begin{equation}
\label{eq-fwd-CLF-common}
V = 
\rho^2+ \zeta^2 + q^2\Gamma^2\,, \quad  q = \sqrt{k_1/k_3} 
\,,
\end{equation}
where the forwarding transformation $\zeta$ is defined below.
\paragraph{GloFo controller ($\mathcal{T}$)} Much like Genova~\eqref{eq:angular_velocity_genova}, the globally asymptotically stabilizing on $\mathcal{T}$ forwarding controller (GloFo, pronounced 'glofo') is given by 
\begin{equation}
\label{eq:GloFo}
\fbox{$\tilde\omega = k_2\gamma + k_3\mbox{sinc}(2\gamma) \zeta
$}
\end{equation}
and has the associated CLF~\eqref{eq-fwd-CLF-common} where $\Gamma = \gamma$ is as in \eqref{eq:T_Delta_Gamma}
and the forwarding transformation reads as
\begin{equation}
    \zeta = \delta + \dfrac{k_1}{2k_2}\mbox{Si}(2\gamma)\,,
\end{equation}
where $\mbox{Si}(a)$ is the sine integral function defined as
\begin{align}
\mbox{Si}(a) = \int_0^{a}\mbox{sinc}(\sigma){\rm d} \sigma
\,. \label{eq:Si_function}
\end{align}

\paragraph{BoFo ($\mathcal{T}_1$)} Similar to BoLSA, the forwarding controller on $\mathcal{T}_1$ is given by 
\begin{equation}
\label{eq:BoFo}
\fbox{$\tilde\omega = k_2\sin\gamma + k_3\dfrac{\cos\gamma}{\left(1 + \tan^2 \dfrac{\gamma}{2}\right)^2}\zeta
$}
\end{equation}
and has the associated CLF~\eqref{eq-fwd-CLF-common}
where $\Gamma = \tan\frac{\gamma}{2}$ is as in \eqref{eq:T1_Delta_Gamma} and the forwarding transformation is
\begin{equation}
    \zeta = \delta + \dfrac{k_1}{k_2}\sin\gamma \,.
\end{equation}
Much like BoLSA, this controller is bounded in the LoS angle $\gamma$ and hence referred to as the Bounded LoS angle by Forwarding (BoFo, pronounced 'bofo'). BoFo's region of attraction is the state space \eqref{eq:ss_S1} and the unicycle never fully turns its back to the goal position while parking as long as the vehicle's initial heading is not exactly opposite to the target. 


\subsection{Backstepping controllers}
Consistent with the backstepping perspective, regard the system \eqref{eq:unicycle_polar_closed_loop-Gv-3-pass} as a double integrator chain, with the sinusoidal nonlinearity acting between the two integrators and limiting the magnitude of the input that acts on $\dot\delta$. Such a nonlinear integrator chain provides an opportunity for an unconventional application of the backstepping method. The next two controllers are designed with such an approach. The CLFs are given in both cases by
\begin{equation}
\label{eq-bkst-CLF-common}
V = 
\rho^2+ \Delta^2 + q^2z^2\,, \quad q = \sqrt{k_1/k_3}
\end{equation}
with the backstepping transformation
\begin{equation}
z = \gamma + \frac{1}{2} \arctan(2k_2 \Delta).
\end{equation}
 The backstepping controllers employ the continuous, bounded function
\begin{equation}
\begin{aligned}[b]
&\psi(z,\gamma) = \frac{\sin(2z-2\gamma) +\sin(2\gamma)}{2z}
\\
 &= \frac{1}{\sqrt{1 + 4k_2^2 \Delta^2}} \left(\,\mbox{sinc}(2z)
+ 2k_2 \Delta \frac{1 - \cos(2z)}{2z}\right)
 \label{eq:psi_definition}
\end{aligned}
\end{equation}
which has the property that $\psi(0,\gamma)  = \cos(2\gamma)$ and, in particular, that $\psi(0,0) = \max\psi(z,\gamma)= \psi(0,n\pi) = 1$, for all integer $n$. We will explain the role of this function in \eqref{eq-bkst-step1}. If, instead of the ``unconventional, bounded integrator chain'' \eqref{eq:unicycle_polar_closed_loop-Gv-3-pass} one had $\dot\delta = k_1 \gamma$, the function $\psi$ would be simply $\psi(z,\gamma)\equiv 1$, namely, the virtual input coefficient. 
Both backstepping controllers take the same general form, given by
\begin{equation}
\tilde\omega = k_4 z +  \left( \frac{k_1k_2 \sin(2\gamma)}{2(1+4k_2^2 \Delta^2)} + k_3 \psi(z,\gamma)\Delta \right)\Delta^{\prime}.
\label{eq:backstepping_controllers}
\end{equation}

\paragraph{GloBa controller ($\mathcal{T}$)} The globally asymptotically stabilizing backstepping controller on $\mathcal{T}$ (GloBa, pronounced 'globa') is obtained by substituting the expression for $\Delta$ from \eqref{eq:T_Delta_Gamma} into \eqref{eq:backstepping_controllers}, yielding
\begin{equation}
\label{eq-bkst-1}
\fbox{$\tilde\omega =  k_4 z + \displaystyle \frac{k_1}{2}\frac{k_2}{1+4k_2^2 \delta^2}\sin(2\gamma) + k_3 \psi( z, \gamma)\delta$}
\end{equation}
and has the associated strict CLF \eqref{eq-bkst-CLF-common} where $\Delta = \delta$ is as in \eqref{eq:T_Delta_Gamma} which is the same CLF reported by Restrepo et al.~\cite{restrepo2020leader}, though their controller differs substantially from \eqref{eq:backstepping_controllers}. 
\paragraph{BAR-FLi controller ($\mathcal{T}_2$)} Finally, we also present a backstepping controller on $\mathcal{T}_2$ which is obtained by substituting the expression for  \eqref{eq:T2_Delta_Gamma} in \eqref{eq:backstepping_controllers}, yielding
\begin{equation}
\label{eq-bkst-3}
\fbox{$
\begin{aligned}
\tilde\omega &= k_4 z + \frac{k_1}{2} \frac{k_2 \left(1 + \tan^2 \frac{\delta}{2}\right)}{
1 + 16k_2^2 \tan^2\displaystyle\frac{\delta}{2}
} \sin(2\gamma)\\[4pt]
&\quad + 2k_3\psi(z,\gamma)
\left(1 + \tan^2\frac{\delta}{2}\right) \tan\frac{\delta}{2}
\end{aligned}
$}
\end{equation}
and has the associated CLF \eqref{eq-bkst-CLF-common} where $\Delta = \tan \frac{\delta}{2}$ is as in \eqref{eq:T2_Delta_Gamma}.
This backstepping feedback's region of attraction is the same as BoPa~\eqref{eq:Bopa_controller}, i.e., the state space \eqref{eq:ss_S2}
and analogously maintains $|\delta(t)|<\pi$, implying that the unicycle never crosses the line in front of the target. Hence, we refer to the controller as the Backstepping to Avoid Running across Front Line (BAR-FLi, pronounced 'Bar Fly') controller.

\subsection{The barrier CLFs and ``nearly global'' feedbacks}\label{sec:barrier}

\paragraph{Barrier CLFs} The Lyapunov functions \eqref{eq:V_CLF_polynomial_summary} and \eqref{eq-fwd-CLF-common} with \eqref{eq:T1_Delta_Gamma} 
blow up at $\gamma=\pm \pi$ and the Lyapunov functions 
\eqref{eq:CLF_ABounD_summary} and
\eqref{eq-bkst-CLF-common} with \eqref{eq:T2_Delta_Gamma} blow up at $\delta=\pm \pi$.  Similarly, the Lyapunov function \eqref{eq:CLF_ABounD_summary} with \eqref{eq:T3_Delta_Gamma} blows up on the boundary of $\mathcal{T}_3$, the square $\{\delta=\pm\pi\}\cup \{\gamma=\pm\pi\}$. 
Such Lyapunov functions are called ``barrier Lyapunov functions'' and they ensure the invariance of the sub-level sets where $|\gamma|  < \pi$ and $|\delta| < \pi$, respectively, are bounded. 
It is convenient to note that, using the definitions in Table~\ref{tab:inverse_polar_coordinates} and 
$\xi = - \rho\cos\gamma = x \cos\theta +y\sin\theta$,  
$\eta = \rho \sin\gamma= x \sin\theta -y\cos\theta 
$,
we have that
\begin{align}
\tan\frac{\delta}{2} &= -\frac{y}{\sqrt{x^2+y^2}-x}=-\frac{\sqrt{x^2+y^2}+x}{y}\label{eq:tan_delta_cart}
\\
\tan\frac{\gamma}{2} &= \frac{x\sin\theta-y\cos\theta}{\sqrt{x^2+y^2}-x\cos\theta-y\sin\theta}
= \frac{\rho+\xi}{\eta}\label{eq:tan_gamma_cart}
\,,
\end{align}
and, hence, we can express the barrier  terms $\tan^2\frac{\delta}{2}$, $\tan^2\frac{\gamma}{2}$ in the Lyapunov functions, as well as the metrics \eqref{eq:metric_S_hat} with \eqref{eq:T1_Delta_Gamma}, \eqref{eq:T2_Delta_Gamma} and \eqref{eq:T3_Delta_Gamma}, respectively, in terms of $(x,y,\theta)$. 

\begin{figure}[t]
\centering
\includegraphics[width=0.7\linewidth]{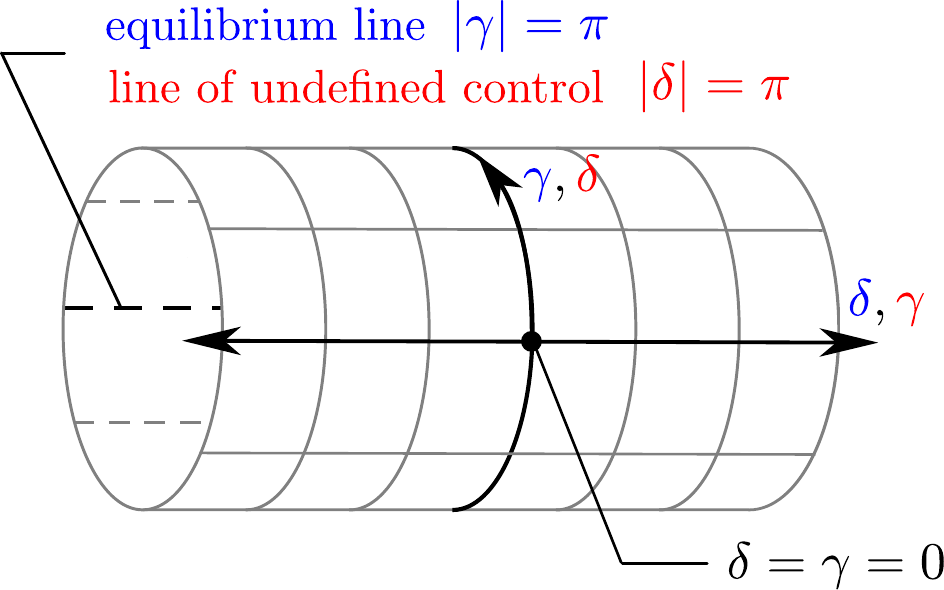}
\caption{Cylindrical state spaces $\mathcal{T}_1$ (blue axes) and $\mathcal{T}_2$ (red axes). For $\mathcal{T}_1$, the set $\{\delta \in \mathbb{R}\}\times\{|\gamma| = \pi\}$ is an equilibrium, while for $\mathcal{T}_2$, $\{|\delta| = \pi\}\times\{\gamma \in \mathbb{R}\}$ is the set where the control is undefined.}
\label{fig:cylinder}
\end{figure}

\paragraph{Feedback actions near and at the barriers} 
The controllers  \eqref{eq-control-bounded-in-gamma}, \eqref{eq:Bopa_controller}, \eqref{eq:BoFo}, \eqref{eq-bkst-3}, 
and \eqref{eq-control-bounded-in-gamma-delta} employing barrier CLFs exclude the measure zero sets $\gamma=\pm \pi$, $\delta=\pm \pi$, and $\{\delta=\pm\pi\}\cup \{\gamma=\pm\pi\}$, from their respective regions of attraction. In \eqref{eq-control-bounded-in-gamma} and \eqref{eq:BoFo} the bounded feedback term $\sin\gamma$ in relation to the {\em barrier} Lyapunov term $\tan^2\frac{\gamma}{2}$ creates unstable equilibria on the set $\gamma=\pm\pi$, while in \eqref{eq:Bopa_controller} and \eqref{eq-bkst-3}, the superlinear term $\left(1+\tan^2\frac{\delta}{2}\right)\tan\frac{\delta}{2}$ grows to infinity to prevent $\delta$ from reaching its boundary, precluding the existence of equilibria on the set $\delta = \pm \pi$. The combined law \eqref{eq-control-bounded-in-gamma-delta} incorporates both. The substantially different behavior among the feedback terms, all of which relate to barrier Lyapunov functions of the form $"\tan^2"$, is because the control input $\omega$ affects $\gamma$ directly but influences $\delta$ only through an integrator (a difference of relative degree of one with respect to $\omega$).

The feedback laws \eqref{eq-control-bounded-in-gamma} and \eqref{eq:BoFo}, at $\gamma = \pm \pi$ and $\delta = 0$ (when the vehicle points away from both the target position and the target heading), drives the system toward the target position $\rho = 0$ through the reverse velocity feedback $v = -k_1\rho$, but does not achieve the target heading. In contrast, for the feedback laws \eqref{eq:Bopa_controller},
\eqref{eq-bkst-3} and \eqref{eq-control-bounded-in-gamma-delta}, starting at $\delta=\pm \pi$, the steering input is $\omega = \pm \infty$, i.e., undefined. This is due to the topological impossibility of achieving global stabilization on 
the cylinder $ S^1\times \mathbb{R}$ for \eqref{eq:Bopa_controller} and \eqref{eq-bkst-3} and on 
the torus $S^1 \times S^1$ for \eqref{eq-control-bounded-in-gamma-delta}.

\paragraph{Topological behavior of controllers} The controllers \eqref{eq-control-bounded-in-gamma}, \eqref{eq:BoFo}, \eqref{eq:Bopa_controller}, and \eqref{eq-bkst-3} are defined on the cylindrical state-spaces $\mathcal{T}_1$ and $\mathcal{T}_2$, illustrated in Fig.~\ref{fig:cylinder} and the BAgAl controller \eqref{eq-control-bounded-in-gamma-delta} is defined on the toroidal state-space $\mathcal{T}_3$ illustrated in Fig. \ref{fig:torus}.  We focus our discussion on the toroidal state-space $\mathcal{T}_3$ associated with the BAgAl controller \eqref{eq-control-bounded-in-gamma-delta}, as this covers the essential features of both $\mathcal{T}_1$ and $\mathcal{T}_2$. The  value of BAgAl \eqref{eq-control-bounded-in-gamma-delta} is undefined on the set $\{\rho>0\}\times \{|\delta|=\pi\}\times \{|\gamma|\leq \pi\} \ = \ \{x>0\}\cap \{y=0\}$, the red circle in Figure \ref{fig:torus} which is a part of the boundary of the open state space $\mathcal{T}_3$ and a measure zero subset of the ``full configuration space'', i.e., $\mathcal{W} := \{\rho>0\}\times \{|\delta|\leq\pi\}\times \{|\gamma|\leq \pi\} \ = \ \{x^2+y^2 > 0\}$. The set $\{|\delta|<\pi\}\times \{|\gamma|= \pi\}$, the blue circle in Fig~\ref{fig:torus} which is also a part of the boundary of the open state space $\mathcal{T}_3$ and a measure zero subset of the torus $\{|\delta|\leq\pi\}\times \{|\gamma|\leq \pi\}$, is an equilibrium set of \eqref{eq:unicycle_polar_closed_loop-Gv-3-pass}, \eqref{eq-control-bounded-in-gamma-delta}, on which the vehicle faces away from the positional target but may be anywhere except $x>0,y=0$. The set $\{\rho>0\}\times \{|\delta|<\pi\}\times \{|\gamma|= \pi\} \ = \ \{x<0\}\cup \{y\neq 0\}$, which is of measure zero in $\mathcal{W}$, consists of straight-line trajectories of \eqref{eq:unicycle_polar_closed_loop-Gv-1}, \eqref{eq-basic-v-control}, \eqref{eq-omega-general}, \eqref{eq-control-bounded-in-gamma-delta} which ``back up'' to the target position with $\theta\neq 0$. From the entire configuration space $\mathcal{W}$, which is a Cartesian product of positive distance and the torus $\mathcal{T}_3$, only the measure zero subset $\{\rho>0\}\times (\{|\delta|=\pi\}\cup \{|\gamma|= \pi\})$, which are the intersecting blue and red circles on the torus in Fig.~\ref{fig:torus}, is excluded from the region of attraction of $\rho=\delta=\gamma=0$. 

\begin{figure}[t]
\centering
\includegraphics[width=0.75\linewidth]{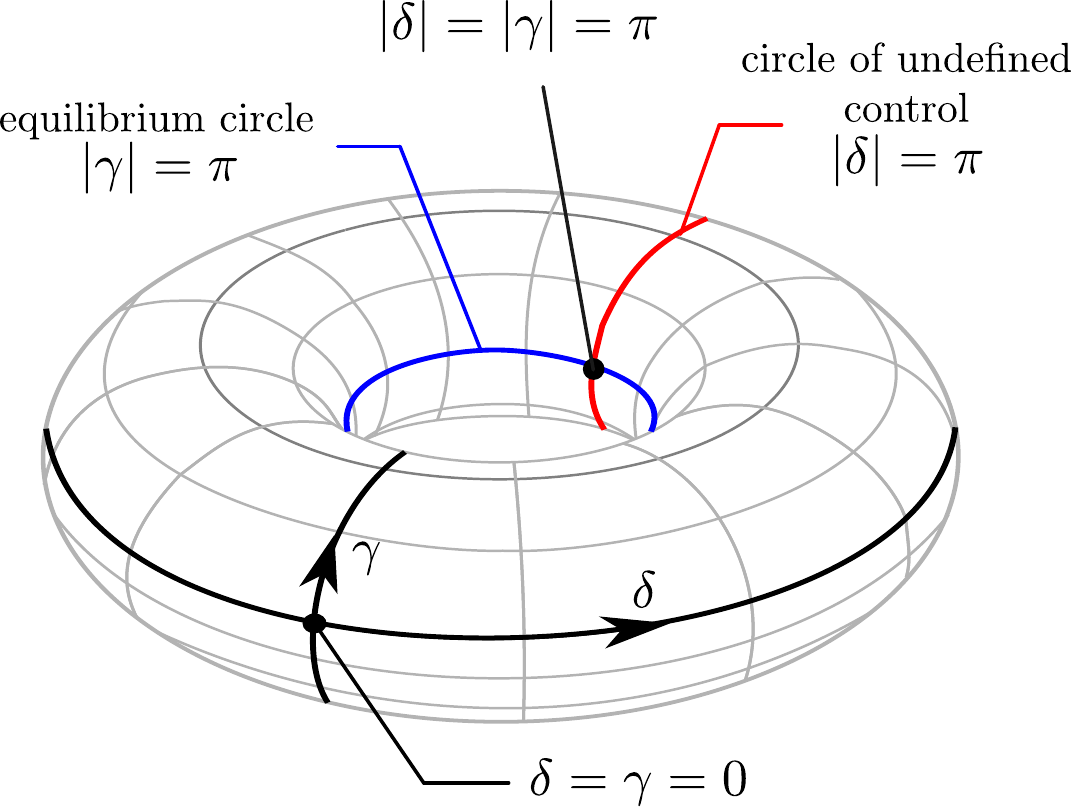}
\caption{Toroidal state-space $\mathcal{T}_3$ with the undefined control set $\{|\delta|=\pi\}\times \{|\gamma|\leq \pi\}$ (red) and the equilibrium set $\{|\delta|<\pi\}\times \{|\gamma|= \pi\}$ (blue).}
\label{fig:torus}
\end{figure}

\section{Composite Lyapunov Functions and CLFs}
\label{sec:composite}
We have generated a multitude of CLFs in the state $(\rho,\delta,\gamma)$. These CLFs are built ``modularly,'' since the system \eqref{eq:unicycle_polar_closed_loop-Gv-3-pass} with $\tilde\omega(\delta,\gamma)$ is independent of $\rho$ and, hence, the CLF for the $(\delta,\gamma)$-subsystem, denoted $V_{\delta\gamma}(\delta,\gamma)$, depends only on  $(\delta,\gamma)$, whereas the CLF for \eqref{eq:unicycle_polar_closed_loop-Gv-2}
naturally depends only on $\rho$, with a nonzero $\gamma$ regarded simply as a disturbance. 
The following proposition gives examples of modularly designed CLFs $V(\rho,\delta,\gamma) = \mathcal{V}(\rho^2,V_{\delta\gamma})$, where, due to the scalar nature of \eqref{eq:unicycle_polar_closed_loop-Gv-2}, its CLF is taken, without loss of generality, as $\rho^2$. 

\begin{proposition}[{Composite Lyapunov functions}]\em 
\label{prop:composite_Lyap_function}
Consider 
any continuously differentiable function $(\delta,\gamma) \mapsto V_{\delta\gamma}$ where $(\delta,\gamma)$ belong to any state space $\hat{\mathcal{T}}\in\{\mathcal{T},\mathcal{T}_1,\mathcal{T}_2,\mathcal{T}_3\}$
and  such that 
$ \alpha_1(|(\delta, \gamma)|_{\hat{\mathcal{T}}}) \le V_{\delta\gamma}(\delta,\gamma) \le \alpha_2(|(\delta, \gamma)|_{\hat{\mathcal{T}}})$, where $\alpha_1, \alpha_2$ are class $\mathcal{K}_{\infty}$ functions. 
Let $(\delta,\gamma) \mapsto\tilde\omega$ be a continuous function such that 
\begin{equation*}
\dot V_{\delta\gamma}(\delta,\gamma) := \dfrac{\partial V_{\delta\gamma}}{\partial\delta}
\frac{k_1\sin(2\gamma)}{2}  -\dfrac{\partial V_{\delta\gamma}}{\partial\gamma}\tilde\omega(\delta,\gamma) \leq -\alpha_{\delta\gamma}(V_{\delta\gamma}) \label{eq:V_dot_delta_gamma}
\end{equation*}
for some class $\mathcal{K}$ function $\alpha_{\delta\gamma}$. 
Then, for {\em any} function $(r,s) \mapsto \mathcal{V}$ such that 
\begin{enumerate}
\item $\mathcal{V}(0,0)=0$ and $\mathcal{V}(r, s)>0$ if $r > 0 $ or $s > 0$, \label{it:pos_def_composite_Lyap_function} 
\item $\lim_{r+s\rightarrow\infty} \mathcal{V}(r,s) = \infty$, \label{it:the_other_condition}
\item $\dfrac{\partial \mathcal{V} }{\partial r }(r,s)>0$ and $\dfrac{\partial \mathcal{V} }{\partial s}(r,s)>0$ if $r>0$ or $s > 0$, \label{it:positive_partials}
\end{enumerate}
for both `composite' Lyapunov functions $V(\rho,\delta,\gamma) = \mathcal{V}\left(\rho^2,V_{\delta\gamma}(\delta,\gamma)\right)$ and $V(\rho,\delta,\gamma) = \mathcal{V}\left(V_{\delta\gamma}(\delta,\gamma), \rho^2 \right)$ there exist respective (distinct) 
triplets of functions $\bar\alpha_1,\bar\alpha_2\in\mathcal{K}_\infty$, $\alpha\in\mathcal{K}$
such that,  for all $(\rho,\delta,\gamma)$ in the state space $\hat{\mathcal{S}}\in\{\mathcal{S},\mathcal{S}_1,\mathcal{S}_2,\mathcal{S}_3\}$, it holds that  
$ \bar\alpha_1(|(\rho,\delta, \gamma)|_{\hat{\mathcal{S}}}) \le V(\rho,\delta,\gamma) \le \bar\alpha_2(|(\rho,\delta, \gamma)|_{\hat{\mathcal{S}}})$ and 
\begin{equation}
\begin{aligned}
\dot{V}(\rho,\delta,\gamma) := 
-\dfrac{\partial V }{\partial \rho}
k_1\rho^2\cos^2\gamma   
&+ \dfrac{\partial V}{\partial \delta}
\frac{k_1 \sin(2\gamma)}{2} \\
&-\dfrac{\partial V}{\partial\gamma}\tilde\omega(\delta,\gamma) \nonumber
\leq -\alpha(V)\,. 
\end{aligned}
\label{eq:V_dot_rho_delta_gamma}
\end{equation}
\end{proposition}

\begin{proof}
Consider first the composite Lyapunov function $V(\rho,\delta,\gamma) = \mathcal{V}(\rho^2, V_{\delta \gamma}(\delta,\gamma))$, which evidently satisfies $ \bar\alpha_1(|(\rho,\delta, \gamma)|_{\hat{\mathcal{S}}}) \le V(\rho,\delta,\gamma) \le \bar\alpha_2(|(\rho,\delta, \gamma)|_{\hat{\mathcal{S}}})$. 
Its time derivative along \eqref{eq:unicycle_polar_closed_loop-Gv-2} and \eqref{eq:unicycle_polar_closed_loop-Gv-3-pass} is given by 
\begin{equation}
\hspace*{-0.32cm}
\begin{aligned}[b]
&\dot{V}(\rho,\delta,\gamma) := 
-\dfrac{\partial \mathcal{V} }{ \partial r}( \rho^2, V_{\delta \gamma}(\delta,\gamma))
k_1\rho^2\cos^2\gamma  \\ &+ \dfrac{\partial \mathcal{V}}{\partial s} ( \rho^2, V_{\delta \gamma}(\delta,\gamma))\left[\dfrac{\partial V_{\delta \gamma}}{\partial \delta}
\frac{k_1}{2} \sin(2\gamma) -\dfrac{\partial V_{\delta \gamma}}{\partial\gamma}\tilde\omega(\delta,\gamma) \right]\,. \label{eq:V_dot_rho_delta_gamma_1}
\end{aligned}
\end{equation}
Taking into account \eqref{eq:V_dot_delta_gamma} and the positivity of the partial derivatives from property~\ref{it:positive_partials}), we have 
\begin{align}
\dot{V}(\rho,\delta,\gamma) &\le -\dfrac{\partial \mathcal{V} }{ \partial r}( \rho^2, V_{\delta \gamma}(\delta,\gamma))
k_1\rho^2\cos^2\gamma\nonumber\\
&\qquad \qquad - \dfrac{\partial \mathcal{V}}{\partial s} ( \rho^2, V_{\delta \gamma}(\delta,\gamma)) \alpha_{\delta \gamma}(V_{\delta \gamma}), \label{eq:V_dot_rho_delta_gamma_1_upper_bound}
\end{align}
which is negative whenever $\rho > 0$ or $(\delta, \gamma) \ne (0,0)$, and zero for $\rho = \delta = \gamma = 0$. By \cite[Lemma 4.3]{khalil_nonlinear_2002}, there exists a class $\mathcal{K}$-function $\alpha$ such that \eqref{eq:V_dot_rho_delta_gamma} holds. 
%
For the composite Lyapunov function $V(\rho,\delta,\gamma) = \mathcal{V}( V_{\delta \gamma}(\delta,\gamma), \rho^2)$, 
an identical argument establishes the result. 
\end{proof}

\begin{corollary} 
\label{cor:example_composite_CLFs}
    The  functions
 $\mathcal{V}(r,s)= r+s$, $\mathcal{V}(r,s)= \ln(1+r) + s$, $\mathcal{V}(r,s)= {\rm e}^{r} - 1 + s$ , $\mathcal{V}(r,s)= (1+r){\rm e}^{s} -1$, $\mathcal{V}(r,s) = r+s+rs$, $\mathcal{V}(r,s) = \cosh(r) + s - 1 $, and $\mathcal{V}(r,s)= \sqrt{1+r}+\sqrt{1+s} - 2$, satisfy properties \ref{it:pos_def_composite_Lyap_function}), \ref{it:the_other_condition}), and \ref{it:positive_partials}) of Proposition \ref{prop:composite_Lyap_function}.
\end{corollary}


The composite Lyapunov functions in Prop.~\ref{prop:composite_Lyap_function} are CLFs for \eqref{eq:unicycle_polar_closed_loop-Gv-1}, in accordance with the following definition. 
\begin{definition}[CLF for the unicycle \eqref{eq:unicycle_polar_closed_loop-Gv-1}]
\label{def-CLF}
A continuously differentiable function $(\rho,\delta,\gamma)\mapsto V$ 
is a \textit{control Lyapunov function} (CLF) with respect to \eqref{eq:unicycle_polar_closed_loop-Gv-1} if it has the following properties.
\begin{enumerate}
\item There exist class $\mathcal{K}$ functions  $(\bar\alpha_1,\bar\alpha_2)$ such that,  for all $(\rho,\delta,\gamma)$ in 
$\Sigma = \{\rho\geq 0\}\times \hat{\mathcal{T}}$, where $\hat{\mathcal{T}}\in\{\mathcal{T},\mathcal{T}_1, \mathcal{T}_2, \mathcal{T}_3\}$, 
$ \bar\alpha_1(|(\rho,\delta, \gamma)|_{\hat{\mathcal{S}}}) \le V(\rho,\delta,\gamma) \le \bar\alpha_2(|(\rho,\delta, \gamma)|_{\hat{\mathcal{S}}})$, where $\hat{\mathcal{S}} = \{\rho> 0\}\times \hat{\mathcal{T}}$. \label{it:clf_prop_1}
\item There exists $(v/\rho,\omega)\in \mathbb{R}^2$ such that
\begin{equation*}
\left[-\dfrac{\partial V}{\partial\rho}\rho\cos\gamma + \left(\dfrac{\partial V}{\partial\delta}+\dfrac{\partial V}{\partial\gamma}\right)\sin\gamma \right]\dfrac{v}{\rho} - \dfrac{\partial V}{\partial\gamma}\omega < 0 \,,
\end{equation*}
for all $(\rho,\delta,\gamma)\neq (0,0,0)$ in $\Sigma$. \label{it:clf_prop_2}
\end{enumerate}
\end{definition}


\section{The ``Genova Controller''}

\label{sec-Genova}
The system 
\eqref{eq:unicycle_polar_closed_loop-Gv-1} in closed-loop with  \eqref{eq-basic-v-control}, \eqref{eq-omega-general},  \eqref{eq:angular_velocity_genova} is
\begin{subequations}
\label{eq:unicycle_polar_closed_loop}
\begin{align}
\dot{\rho} &=  -k_1 \rho \cos^2(\gamma),  \\
\dot{\delta} &=  k_1 \frac{\sin(2\gamma)}{2}, \label{eq:delta_dot_genova}\\
\dot{\gamma} &= -k_2 \gamma -
k_3 \,\mbox{sinc}(2\gamma)
\delta.\label{eq:gamma_dot_genova}
\end{align}
\end{subequations}
Note that we refer to the feedback laws \eqref{eq-basic-v-control} and \eqref{eq:angular_velocity_genova} as the ``Genova controller'' \cite{aicardi1995} since three of its four coauthors were affiliated with the University of Genova either at the time the paper was published or at a later time. 

\subsection{Global asymptotic stability with a non-strict CLF}

Consider the energy function \eqref{eq:U_common_passivity} where $\Delta$ and $\Gamma$ are defined as in \eqref{eq:T_Delta_Gamma}.
Its time derivative along \eqref{eq:unicycle_polar_closed_loop} is  
\begin{equation}
\dot{U} = - 2 k_2 q^2 \gamma^2. \label{eq:dot_U_polynomial}
\end{equation}
In view of this,
the CLF proposed in \cite{aicardi1995} takes the form
\begin{equation}
V_{\rm G} = \rho^2 + k_3 U \,. \label{eq:CLF_aicardi}
\end{equation}
We refer to it as the ``Genova'' (non-strict) CLF.
Its derivative along the trajectories of \eqref{eq:unicycle_polar_closed_loop} is given by
\begin{align}
\dot{V}_{\rm G} = -2 k_1 \rho^2\cos^2(\gamma) - 2 k_1k_2\gamma^2. \label{eq:V_dot_aicardi}
\end{align}
The function $\dot{V}_{\rm G}(\rho,\gamma,\delta)$ is only negative semi-definite, since the expression contains no $\delta$-dependent term, but $\rho(t),\gamma(t) \rightarrow 0$ as $t\rightarrow \infty$ and $\rho(t)\equiv\gamma(t)\equiv 0$ implies that $\delta(t)\equiv 0$. As a result, by 
the Barbashin-Krasovskii theorem, the origin of \eqref{eq:unicycle_polar_closed_loop} is globally asymptotically stable.

\subsection{Strictifying and generalizing the Genova CLF}
\label{sec-strictify}

\begin{theorem}
\label{thm:unicycle_CLF_polynomial}
Consider the system 
\eqref{eq:unicycle_polar_closed_loop-Gv-1} in closed-loop with \eqref{eq-basic-v-control}, \eqref{eq-omega-general},  \eqref{eq:angular_velocity_genova}, 
with $k_1, k_2, k_3 > 0$ such that
$k_1 k_3\geq k_2^2$. 
The point $\rho = \delta = \gamma = 0$ is GAS on $\mathcal{S}$ in accordance with Def.~\ref{def-our-GAS}.
Furthermore, all the composite Lyapunov functions $V(\rho, \delta, \gamma) = \mathcal{V}(\rho^2, V_{\delta \gamma})$ and $V(\rho, \delta, \gamma) = \mathcal{V}(V_{\delta \gamma},\rho^2)$, for all functions $\mathcal{V}$ satisfying the conditions in Proposition \ref{prop:composite_Lyap_function}, and with $V_{\delta \gamma}$ defined as
\begin{equation}
\label{eq:CLF_polynomial_delta_gamma}
V_{\delta \gamma}(\delta,\gamma) =  k_3  \left(1+
\frac{2q^2+U}{2qk_2}\right)U
+   (\delta+q\gamma)^2 \,,
\end{equation}
where 
\begin{equation}
U=\delta^2+ q^2 \gamma^2\,, \quad q = \sqrt{k_1/k_3}, \label{eq:U_genova}  
\end{equation}
are (globally) strict CLFs for \eqref{eq:unicycle_polar_closed_loop-Gv-1} w.r.t. the input pair $(v/\rho,\omega)$ in the sense of Def. \ref{def-CLF}. 
\end{theorem}

\begin{proof}
To construct strict CLFs, we seek to improve upon $V_{\rm G}$ in \eqref{eq:CLF_aicardi}. Notably, damping in $\gamma$ is already introduced through the $U$ term, as reflected by the $-\gamma^2$ term in its time derivative. We aim to introduce a similar damping effect in $\delta$ to achieve a strict Lyapunov function.
Define
\begin{equation}
    \Pi = z^2, \quad  z = \delta + q \gamma \,.\label{eq:Pi_polynominal}
\end{equation}
Then, from (\ref{eq:unicycle_polar_closed_loop}) we obtain
\begin{equation}
\dot{z}
= -k_2 z + k_2 \delta - k_3 q \,\mbox{sinc}(2\gamma)
\left(\delta-q \gamma \right).
\end{equation}
Taking the time derivative of \eqref{eq:Pi_polynominal} yields 
\begin{equation}
\hspace*{-0.2cm}
\dot{\Pi} = -2k_2 z^2 + 2k_2 \delta z
- 2k_3 q \,\mbox{sinc}(2\gamma)
\delta^2 +  k_1 q \sin(2\gamma) \gamma \,,\label{eq:dot_Pi_polynomial}    
\end{equation} 
where $k_3 q^3 = k_1 q$. Since 
$z \delta \le \frac{z^2}{4} + \delta^2$ and
$\frac{\sin(2\gamma) \gamma}{2}  \le  \gamma^2$, we get
\begin{equation}
\dot{\Pi} 
\le - \frac{3}{2} k_2 \Pi +  2 k_1 q  \gamma^2
+  2 k_2 \delta^2 \left( 1 -   q \frac{k_3}{k_2} \,\mbox{sinc}(2\gamma)
\right) 
\label{ddtpi}.
\end{equation}
Using \eqref{eq:bound1} in Lemma \ref{lemma:sinterm_upperbound}, with $ k_1 k_3 \ge k_2^2$, we upper bound \eqref{ddtpi} as
\begin{equation}
\label{eq-Pidot-problem}
\dot{\Pi} \leq -  \frac{3}{2} k_2 \Pi +2 k_1 q  \gamma^2 + 2 k_3 q  \delta^2 \gamma^2 \,.
\end{equation}
We observe that $\Pi$  introduces a damping effect in $\delta$. However, it also produces undesirable terms involving $\gamma^2$ and $\delta^2 \gamma^2$. To eliminate their effect, consider \eqref{eq:U_genova} and \eqref{eq:dot_U_polynomial}, and note that the time derivative of $\frac{k_1}{k_2 q}U$ gives $\frac{k_1}{k_2 q}\dot{U} = -2k_1 q \gamma^2$, whereas the time derivative of $\frac{k_3}{2k_2 q}U^2$ gives
\begin{equation}
\frac{k_3 }{2 k_2 q}\ddt{U^2}  = 
- \frac{ k_3 q }{ 2} 
U  \gamma^2
=-2 k_1 q \gamma^2 - 2 k_3 q  \delta^2 \gamma^2\,.
\end{equation}
To cancel the positive terms in \eqref{eq-Pidot-problem}, we augment  $\Pi$ as
\begin{align}
\Pi_1= \Pi+ \frac{k_1}{k_2 q} U +  \frac{k_3}{2k_2 q}U^2 \,,
\label{eq:Pi_1_genova}
\end{align}
and obtain
\begin{align}
\dot\Pi_1\leq - \frac{3}{2} k_2 \Pi -   2 k_1 q V_0^2. \label{eq:V_dot_SD}
\end{align}
In conclusion, from \eqref{eq:U_genova} and \eqref{eq:Pi_1_genova}, we construct a Lyapunov function for the $(\delta,\gamma)$-subsystem as $V_{\delta \gamma} = k_3 U + \Pi_1$ which is equivalent to \eqref{eq:CLF_polynomial_delta_gamma}, 
for which there exist class $\mathcal{K}_{\infty}$ functions $\alpha_1, \alpha_2$ such that $\alpha_1(\abs{(\delta,\gamma)}_{\mathcal{T}}) \le V_{\delta \gamma }(\delta, \gamma) \le \alpha_2(\abs{(\delta,\gamma)}_{\mathcal{T}})$. The time-derivative of \eqref{eq:CLF_polynomial_delta_gamma} along the trajectories of \eqref{eq:unicycle_polar_closed_loop-Gv-3-pass}, \eqref{eq:angular_velocity_genova} is 
\begin{equation}
\dot{V}_{\delta \gamma} = - 2 k_1k_2\gamma^2 - \frac{3}{2} k_2 (\delta + q \gamma)^2 -   2 k_1 q \gamma^4\,,
\end{equation}
which is negative for all $(\rho,\delta, \gamma) \ne (0,0,0)$ in $\{\rho \ge 0\} \times \mathcal{T}$, and 
based on \cite[Lemma 4.3]{khalil_nonlinear_2002},  
there exists 
$\alpha_3\in\mathcal{K}$, such that, $\dot{V}_{\delta \gamma} \le - \alpha_3 \circ \alpha_2^{-1}
(V_{\delta \gamma})$, where $
\alpha_3 \circ \alpha_2^{-1}\in\mathcal{K}$. 
The Lyapunov function $V_{\delta \gamma}$ satisfies the assumptions of Proposition \ref{prop:composite_Lyap_function}, which implies that the composite Lyapunov functions $V(\rho, \delta, \gamma) = \mathcal{V}(\rho^2, V_{\delta \gamma})$ and $V(\rho, \delta, \gamma) = \mathcal{V}(V_{\delta \gamma},\rho^2)$ satisfy $\dot{V}(\rho,\delta,\gamma) \le - \alpha(V)$, and hence are strict CLFs for  \eqref{eq:unicycle_polar_closed_loop-Gv-1} in the sense of Def. \ref{def-CLF}. 
Furthermore, $\dot{V}(\rho,\delta,\gamma) \le - \alpha(V)$ implies existence of  $\beta\in\mathcal{KL}$  such that $V(t) \le \beta(V_0, t)$, 
for all $t\geq 0$ according to \cite[Lemma 4.4]{khalil_nonlinear_2002}. Considering the fact, from Proposition \ref{prop:composite_Lyap_function}, that the composite Lyapunov functions are bounded as $ \bar\alpha_1(|(\rho,\delta, \gamma)|_{\mathcal{S}}) \le V(\rho,\delta,\gamma) \le \bar\alpha_2(|(\rho,\delta, \gamma)|_{\mathcal{S}})$, from  $V \le \beta(V_0) $, we have that     
\begin{equation}
\abs{(\rho,\delta,\gamma)}_{\mathcal{S}} \le \bar{\alpha}_1^{-1}(\beta( \bar{\alpha}_2(\abs{(\rho_0, \delta_0,\gamma_0)}_{\mathcal{S}}),t)) \,, \label{eq:KL-estimate}
\end{equation}
where $\bar{\alpha}_1^{-1}(\beta( \bar{\alpha}_2(r),t))$ is class $\mathcal{KL}$ in $(r,t)$.
Thus, 
$\rho =  \delta = \gamma = 0$ is GAS on $\mathcal{S}$ in accordance to Def. \ref{def-our-GAS}.
\end{proof}

Theorem \ref{thm:unicycle_CLF_polynomial} establishes global asymptotic stability in the polar variables $(\rho,\delta,\gamma)$. One should not expect that GAS also follows for the closed-loop system in the Cartesian variables $(x,y,\theta)$. However,  attractivity of the point $(x,y,\theta)=(0,0,0)$ does follow. 

\begin{corollary}
\label{cor-attractivity}
\textit Consider the system 
\eqref{eq:unicycle_cartesian} in closed-loop with \eqref{eq-basic-v-control}, \eqref{eq-omega-general},  \eqref{eq:angular_velocity_genova}.
For all initial conditions $(x_0,y_0,\theta_0)\in\mathbb{R}^3$ such that $x_0^2 + y_0^2>0$, one has 
\begin{align*}
&\abs{x(t)} + \abs{y(t)} + \abs{\theta(t)}\nonumber\\
&\le  {\beta}(\abs{x_0} + \abs{y_0} + \abs{\theta_0}  + \abs{{\rm atan2}(y_0,x_0) + \pi},t), \quad \forall t \ge 0 \label{eq:KL-estimate_xytheta}
\end{align*}
where $\beta\in\mathcal{KL}$  is defined as  $\beta(r,t) := \sqrt{2}\bar{\alpha}_1^{-1}(\beta( \bar{\alpha}_2(2 r),t))$.
\end{corollary}

This corollary follows from \eqref{eq:KL-estimate}, with the aid of the inequalities 
$\frac{1}{\sqrt{2}} (\abs{x} + \abs{y} + \abs{\theta}) \le \abs{(\rho,\delta,\gamma)}_{\mathcal{S}} \le \abs{x} + \abs{y} + \abs{\theta}  + 2\abs{\delta}$,  established with the inverse transformations from Table \ref{tab:inverse_polar_coordinates}.

The achievement of attractivity, despite not having stability in the Cartesian coordinates $(x,y,\theta)$, expressed in Corollary~\ref{cor-attractivity} is not a shortcoming of our design method. It is consistent with the result of \cite[Remark 1.6]{coron1994relation}, proven also in \cite[pg. 43]{praly2022fonctions}, that the unicycle, in the Cartesian representation, is impossible to stabilize by static feedback, even if  feedback is permitted to be discontinuous, as is the Genova feedback in the Cartesian coordinates, 
as well as all the other feedback laws in this paper. 
\begin{remark}
    In an independent discovery under review~\cite{BoWang2025further}, submitted after the development of our result,  
a global strict CLF for the closed-loop \eqref{eq:unicycle_polar_closed_loop}
shares some similarities to our CLFs  in Theorem~\ref{thm:unicycle_CLF_polynomial}, but is neither as general nor proven generally composable with $\rho^2$ (as in our Prop. \ref{prop:composite_Lyap_function}). Our further contributions are the three additional CLFs on the state-spaces \eqref{eq:ss_S1}--\eqref{eq:ss_S3}, with which we build a full passivity-based methodology comprising a quartet of designs.
\end{remark}

\section{Passivity-Based Angle-Constraining Controllers}

\subsection{Bounded-in-LoS-Angle (BoLSA) Controller}\label{sec:BoLSA}

\begin{theorem}[BoLSA CLFs] 
\label{thm:unicycle_CLF_BoLSA}
Consider the system \eqref{eq:unicycle_polar_closed_loop-Gv-1} in closed-loop with \eqref{eq-basic-v-control}, \eqref{eq-omega-general} and \eqref{eq-control-bounded-in-gamma}, with  $k_1, k_2, k_3 > 0$ such that 
$k_1 k_3\geq k_2^2$.
The point $\rho = \delta = \gamma = 0$ is GAS on $\mathcal{S}_1$ in accordance with Def. \ref{def-our-GAS}.
Furthermore, all the composite Lyapunov functions $V(\rho, \delta, \gamma) = \mathcal{V}(\rho^2, V_{\delta \gamma})$ and $V(\rho, \delta, \gamma) = \mathcal{V}(V_{\delta \gamma},\rho^2)$, for all functions $\mathcal{V}$ satisfying the conditions in Proposition \ref{prop:composite_Lyap_function}, and with $V_{\delta \gamma}$ defined as
\begin{equation}
\label{eq:CLF_Bolsa_delta_gamma}
V_{\delta \gamma}(\delta, \gamma) =  k_3  \left(1+
\frac{2q^2+U}{2qk_2}\right)U
+   \left (\delta+2q\tan \frac{\gamma}{2} \right)^2\,,
\end{equation}
where
\begin{equation}
U=\delta^2+  4q^2 \tan^2\frac{\gamma}{2} \,, \quad q = \sqrt{k_1/k_3} \,, \label{eq:U_bolsa}
\end{equation} 
are (globally) strict CLFs for \eqref{eq:unicycle_polar_closed_loop-Gv-1} w.r.t. the input pair $(v/\rho,\omega)$ in the sense of Def.~\ref{def-CLF}. 
\end{theorem}

\begin{proof}
We start by studying the feedback law \eqref{eq-control-bounded-in-gamma} and the CLF \eqref{eq:CLF_Bolsa_delta_gamma} for the $(\delta, \gamma)$-subsystem   \eqref{eq:unicycle_polar_closed_loop-Gv-3-pass}. For this, consider the trigonometric indentities:
\begin{eqnarray}
\tan\frac{\gamma}{2} &=& \frac{\sin\gamma}{1+\cos\gamma} \label{eq:trig_identity_1}
\\
\left(\tan\frac{\gamma}{2}\right)^\prime &=& \frac{1}{1+\cos\gamma}  = \frac{1}{2}\left(1+\tan^2\frac{\gamma}{2}\right) \label{eq:trig_identity_2}
\\
\left(\tan^2\frac{\gamma}{2}\right)' &=& \frac{2 }{\sin \gamma} \tan^2 \frac{\gamma}{2} =
2 \frac{\sin\gamma}{(1+\cos\gamma)^2}. \label{eq:trig_identity_3}
\end{eqnarray}
Define
\begin{equation}
\Pi = z^2, \quad  z = \delta + 2 q \tan\frac{\gamma}{2}. 
\end{equation}
Note that the term $4\tan^2 (\gamma/2)$ in \eqref{eq:U_bolsa} is radially unbounded on the interval $(-\pi,\pi)$, namely, it disallows ``winding'' the LoS angle to arbitrary values. 
The time derivative of the Lyapunov expression \eqref{eq:U_bolsa} along the solutions of \eqref{eq:unicycle_polar_closed_loop-Gv-3-pass} is
\begin{equation}
\dot U  =  \frac{8 q^2 \sin\gamma}{(1+\cos\gamma)^2} \left[ k_3 \delta \frac{(1+\cos\gamma)^2}{4} \cos\gamma 
- \tilde{\omega}\right]\,.
\end{equation}
We pick the control as in \eqref{eq-control-bounded-in-gamma}, where we use the fact that $(1+\cos\gamma)^2/4 = 1/(1+\tan^2 \gamma/2)^2$, which gives
\begin{equation}
\dot U  = -2 \frac{k_2  q^2 4 \sin^2\gamma}{(1+\cos\gamma)^2} = - 2 k_2 4 q^2  \tan^2 \frac{\gamma}{2}\,,
\end{equation}
and the dynamics for the LoS angle 
\begin{equation}
\dot\gamma = - k_2 \sin\gamma -  k_3 \delta\frac{(1+\cos\gamma)^2}{4} \cos\gamma  \,.
\end{equation}
Next, we consider the dynamics of the error variable $z$, which are given by
\begin{equation}
\hspace*{-0.2cm}
\dot z = -k_2z + k_2\delta + k_3 q \frac{\cos\gamma(1+\cos\gamma)}{2}\left( 2 q \tan\frac{\gamma}{2}-\delta\right).    
\end{equation}
Then,
\begin{equation}
\begin{aligned}[b]
\dot \Pi = - 2k_2z^2 + 2k_2 z \delta &+ 
4 k_1 q \cos\gamma(1+\cos\gamma) \tan^2\frac{\gamma}{2} \\
&\qquad  -k_3 q \frac{\cos\gamma(1+\cos\gamma)} {2}\delta^2\,.
\end{aligned}
\end{equation}
Since $z\delta \leq  \frac{z^2}{4} + \delta^2$ and $\cos\gamma(1+\cos\gamma) \tan^2\frac{\gamma}{2} \leq  2 \tan^2 \frac{\gamma}{2}
$, we get
\begin{equation}
\hspace*{-0.2cm}
\begin{aligned}[b]
    \dot \Pi \le - \frac{3}{2}k_2 \Pi &+ 2 k_ 1 q 4 \tan^2 \frac{\gamma}{2} \\
&+ 2 k_2 \delta^2\left(1- q \frac{k_3 }{2 k_2} 
\cos\gamma(1+\cos\gamma)\right)\,.
\end{aligned}
\end{equation}
It then follows from \eqref{eq:bound2} in Lemma \ref{lemma:sinterm_upperbound} with $k_1 k_3 \ge  k_2^2$ that 
\begin{equation}
\dot\Pi
\leq - \frac{3}{2} k_2 \Pi + 2 k_1 q  4 \tan^2 \frac{\gamma}{2} +   2 k_3 q \delta^2 4 \tan^2 \frac{\gamma}{2}.
\end{equation}
Taking into account that $\frac{k_3 }{k_2} q \dot U = - 2k_2 q^2 V_0$, and $\frac{k_3}{2q k_2 } \dot{U}^2 = -2k_3 q \delta^2 4 \tan^2 \frac{\gamma}{2} - 2k_1 q \left(4 \tan^2 \frac{\gamma}{2}\right)^2$ and denoting
\begin{equation}
\Pi_1 = \Pi + \frac{k_3 }{k_2} q  U  + \frac{k_3 }{2 q k_2} U^2\,, \label{eq:Pi_1_bolsa}
\end{equation}
we get
\begin{equation}
\dot\Pi_1 \leq - \frac{3}{2}k_2\Pi - 2 k_1 q \left(4 \tan^2 \frac{\gamma}{2}\right)^2\,.
\end{equation}
From \eqref{eq:U_bolsa} and \eqref{eq:Pi_1_bolsa} we construct the Lyapunov function for the $(\delta, \gamma)$-subsystem as $V_{\delta \gamma} = k_3 U + \Pi_1$ which is equivalent to \eqref{eq:CLF_Bolsa_delta_gamma} 
and has the time derivative along the solutions of \eqref{eq:unicycle_polar_closed_loop-Gv-3-pass}, \eqref{eq-control-bounded-in-gamma} such that
\begin{equation}
    \dot V_{\delta \gamma} \le  - 2k_1 k_2 V_0  - \frac{3}{2}k_2 \left(\delta + q \tan \frac{\gamma}{2}\right)^2 - 2k_1 q V_0^2\,,
\label{eq:V_dot_Bolsa}
\end{equation} 
with $V_0 = 4 \tan^2\frac{\gamma}{2}$, which is negative for all $(\rho,\delta, \gamma) \ne (0,0,0)$ in $\{\rho \ge 0\} \times \mathcal{T}_1$. Analogous to the proof of Thm.~\ref{thm:unicycle_CLF_polynomial}, we conclude that $V(\rho, \delta, \gamma) = \mathcal{V}(\rho^2, V_{\delta \gamma})$ and $V(\rho, \delta, \gamma) = \mathcal{V}(V_{\delta \gamma},\rho^2)$ are such that $\dot{V}(\rho,\delta,\gamma) \le - \alpha(V)$, and hence are strict CLFs for \eqref{eq:unicycle_polar_closed_loop-Gv-1}, and  
$\rho =  \delta = \gamma = 0$ is GAS on $\mathcal{S}_1$.
\end{proof}

\subsection{Bounded-in-Polar-Angle (BoPA) Controller}

\begin{theorem}[BoPA CLFs] 
\label{thm:unicycle_CLF_BoPA}
Consider the system \eqref{eq:unicycle_polar_closed_loop-Gv-1} in closed-loop with \eqref{eq-basic-v-control}, \eqref{eq-omega-general} and \eqref{eq:Bopa_controller}, with  $k_1, k_2, k_3 > 0$ such that  
$k_1 k_3\geq k_2^2$.
The point $\rho = \delta = \gamma = 0$ is GAS on $\mathcal{S}_2$ in accordance with Def. \ref{def-our-GAS}.
Furthermore, all the composite Lyapunov functions $V(\rho, \delta, \gamma) = \mathcal{V}(\rho^2, V_{\delta \gamma})$ and $V(\rho, \delta, \gamma) = \mathcal{V}(V_{\delta \gamma},\rho^2)$, for all functions $\mathcal{V}$ satisfying the conditions in Proposition \ref{prop:composite_Lyap_function}, and with $V_{\delta \gamma}$ defined as
\begin{equation}
\label{eq:CLF_BoPa_delta_gamma}
\hspace*{-0.35cm}
V_{\delta \gamma}(\delta,\gamma) =  \tilde{a} \left[ (1+U)^3 -1\right]
+   \left (2\tan\frac{\delta}{2} + q \gamma \right)^2\,,
\end{equation}
where
\begin{equation}
\quad U= 4\tan^2 \frac{\delta}{2}+ q^2 \gamma^2\,,  \quad q = \sqrt{k_1 / k_3} \,, \label{eq:U_Bopa}
\end{equation}
with  $ \tilde{a} = \max\{k_1 q, \sqrt{k_1 k_3}\} / 3k_2 q^2$ 
are (globally) strict CLFs for \eqref{eq:unicycle_polar_closed_loop-Gv-1} w.r.t. the input pair $(v/\rho,\omega)$ in the sense of Def.~\ref{def-CLF}. 
\end{theorem}
\raggedbottom
\begin{proof}
We rewritte \eqref{eq:U_Bopa} as
\begin{equation}
    U(\delta,\gamma) =  U_1 + q^2 U_2 \,, \quad U_1 = 4\tan^2\frac{\delta}{2}\,, \quad U_2 = \gamma^2\,, \label{eq:U_BoPa_proof}
\end{equation}
and define
\begin{equation}
     \Pi = z^2\,,  \quad z = 2\tan\frac{\delta}{2} +  q \gamma.
\end{equation}
Note that $U(\delta,\gamma)$ is radially unbounded on $\mathcal{T}_2$, namely, it disallows ``winding'' of the polar angle to arbitrary values. 
Then, the time derivative of \eqref{eq:U_Bopa} along the solutions of \eqref{eq:unicycle_polar_closed_loop-Gv-3-pass} is
\begin{equation}
\dot U  =  2 q^2 \gamma \left[ 2k_3 {\rm sinc}(2\gamma)\left(1 + \tan^2 \frac{\delta}{2}\right)\tan^2\frac{\delta}{2}-\tilde{\omega}
\right]\,. 
\end{equation}
We pick the control as in \eqref{eq:Bopa_controller} and obtain
\begin{equation}
\dot U  =    - 2 k_2 q^2 \gamma^2 = -2k_2 U_2 \,. \label{eq:U_dot_bopa}
\end{equation}
With the feedback~\eqref{eq:Bopa_controller}, the LoS dynamics are 
\begin{equation}
\dot\gamma = - k_2 \gamma -  2k_3 {\rm sinc}(2\gamma)\left(1 + \tan^2 \frac{\delta}{2}\right)\tan^2\frac{\delta}{2} \,.
\end{equation}
Next, we consider the dynamics of the error $z$, which are 
\begin{equation}
\begin{aligned}[b]
\dot z &= -k_2 z + 2 k_2 \tan \frac{\delta}{2}\\
&- k_3 q {\rm sinc}(2\gamma)\left(1+\tan^2\frac{\delta}{2 }\right)  \left( 2 \tan \frac{\delta}{2 } - q \gamma \right).
\end{aligned}
\end{equation}
Then,
\begin{equation}
\begin{aligned}[b]
\frac{1}{2}\dot \Pi &= - k_2z^2 +  2 k_2 z \tan \frac{\delta}{2} \\
&- 
{\rm sinc}(2\gamma)\left(1+\tan^2\frac{\delta}{2 }\right)  \left( U_1 - q^2 U_2 \right). \label{eq:Pi_dot_BoPa_temp}
\end{aligned}
\end{equation}
Using $2 k_2 z \tan \frac{\delta}{2} \leq  \frac{1}{4}k_2 z^2 + 2 k_2 \tan^2 \frac{\delta}{2}$, $ 1  \le 1+\tan^2\frac{\delta}{2} $ and $2k_2 z \tan \frac{\delta}{2} \leq  \frac{1}{4}k_2 z^2 + 2 k_2 \tan^2 \frac{\delta}{2} \left(1+\tan^2\frac{\delta}{2}\right)$,
we get
\begin{equation}
\hspace*{-0.2cm}
\begin{aligned}[b]
&\frac{1}{2} \dot \Pi \le -  \frac{3}{4} k_2 \Pi +  k_2 U_1 \left( 1 + \tan^2\frac{\delta}{2} \right) \times  \\
& \left(1 - \frac{k_3}{k_2}q {\rm sinc}(2\gamma) \right) +k_1 q \left( 1 + \tan^2\frac{\delta}{2} \right) {\rm sinc}(2\gamma) U_2 \,. \label{eq:Pi_dot_BoPa_temp}
\end{aligned}
\end{equation}
It then follows from \eqref{eq:bound1} in Lemma \ref{lemma:sinterm_upperbound} with $k_1 k_3 \ge  k_2^2$ along with ${\rm sinc}(2\gamma) \gamma^2 \le \gamma^2$ that 
\begin{equation}
\begin{aligned}[b]
   \frac{1}{2}\dot \Pi \le -\frac{3}{4} k_2\Pi + q \left(k_1 + k_3U_1 \right)\left(1 +  \tan^2 \frac{\delta}{2}\right) U_2.
\end{aligned}
  \label{eq:temp_P0_bopa}
\end{equation}
With $a = \max\{k_1q,k_3q\}$ and $1+\tan^2\frac{\delta}{2} \le 1+ 4\tan^2\frac{\delta}{2}$ from
\eqref{eq:temp_P0_bopa}, we obtain 
\begin{equation}
\frac{1}{2} \dot{\Pi} \le -\frac{3}{4}k_2 \Pi + a \left(1 +  U_1\right)^2 U_2. \label{eq:Pi0_temp2_bopa}
\end{equation}
Taking \eqref{eq:U_BoPa_proof} into account, we get 
\begin{align}
\dot{\Pi} &\le -\frac{3}{2}k_2 \Pi + 2a   \left( 1 + 2U_1 + U_1^2\right) U_2 \nonumber \\
&\le -\frac{3}{2}k_2 \Pi + 2a   \left( 1 + 2U - 2q^2 U_2 + U^2 \right) U_2 \nonumber \\
&\le -\frac{3}{2}k_2 \Pi + 2a   \left( 1 + U \right)^2 U_2 - 4a q^2 U_2^2. \label{eq:Pi0_temp3_bopa}
\end{align}
Now, consider the Lyapunov function 
\begin{equation}
(1+U)^3 - 1 = U^3 + 3U^2 + 3U 
\end{equation}
which is positive definite since $U$ is positive definite. Recall, from \eqref{eq:U_dot_bopa}, that $\dot U = -2k_2 q^2 U_2$, then
\begin{align}
\frac{a}{3k_2 q^2}\frac{{\rm d} }{{\rm d} t} \left[ (1+U)^3 - 1\right] = - 2a(1+U)^2 U_2 \,,
\end{align}
which cancels out the positive term in \eqref{eq:Pi0_temp3_bopa}.
Thus, the Lyapunov function for the $(\delta,\gamma)$-subsystem is  \eqref{eq:CLF_BoPa_delta_gamma}
with the time-derivative along the solutions of \eqref{eq:unicycle_polar_closed_loop-Gv-3-pass}, \eqref{eq:Bopa_controller} is such that
\begin{equation}
\dot V_{\delta \gamma} \le  - \frac{3}{2} k_2 \left(\tan \frac{\delta}{2} + q \gamma \right)^2 - 4 a q^2 \gamma^4 \,,
\label{eq:V_dot_bopa}
\end{equation}
which is negative for all $(\rho,\delta, \gamma) \ne (0,0,0)$ in $\{\rho \ge 0\} \times \mathcal{T}_2$. Analogous to the proof of Thm \ref{thm:unicycle_CLF_polynomial},  the Lyap. functions $V(\rho, \delta, \gamma) = \mathcal{V}(\rho^2, V_{\delta \gamma})$ and $V(\rho, \delta, \gamma) = \mathcal{V}(V_{\delta \gamma},\rho^2)$  are strict CLFs for \eqref{eq:unicycle_polar_closed_loop-Gv-1} and 
$\rho =  \delta = \gamma = 0$ is GAS on $\mathcal{S}_2$.
\end{proof}

\subsection{Bounding-Angles Algorithm (BAgAl) Controller}

\begin{theorem}[BAgAl CLFs] 
\label{thm:unicycle_CLF_Bagal}
Consider the system \eqref{eq:unicycle_polar_closed_loop-Gv-1} in closed-loop with \eqref{eq-basic-v-control}, \eqref{eq-omega-general} and \eqref{eq-control-bounded-in-gamma-delta}, with  $k_1, k_2, k_3 > 0$ such that  
$k_1 k_3\geq k_2^2$.
The point $\rho = \delta = \gamma = 0$ is GAS on $\mathcal{S}_3$ in accordance with Def. \ref{def-our-GAS}.
Furthermore, all the composite Lyapunov functions $V(\rho, \delta, \gamma) = \mathcal{V}(\rho^2, V_{\delta \gamma})$ and $V(\rho, \delta, \gamma) = \mathcal{V}(V_{\delta \gamma},\rho^2)$, for all functions $\mathcal{V}$ satisfying the conditions in Proposition \ref{prop:composite_Lyap_function}, and with $V_{\delta \gamma}$ defined as
\begin{equation}
\label{eq:CLF_ABounD_delta_gamma}
\hspace*{-0.25cm}
V_{\delta \gamma}(\delta,\gamma) =  \tilde{a} \left[ (1+U)^3 -1\right]
+   \left (2\tan\frac{\delta}{2} + 2q\tan\frac{\gamma}{2} \right)^2 \,,
\end{equation}
where
\begin{equation}
\quad U= 4\tan^2 \frac{\delta}{2}+ 4q^2 \tan^2 \frac{\gamma}{2}\,, \quad q = \sqrt{k_1 / k_3} \,,
\end{equation}
with $\tilde{a} = \max\{k_1 q, \sqrt{k_1 k_2} \} / 3k_2q^2$, 
are (globally) strict CLFs for \eqref{eq:unicycle_polar_closed_loop-Gv-1} w.r.t. the input pair $(v/\rho,\omega)$ in the sense of Def.~\ref{def-CLF}. 
\end{theorem}

\begin{proof}
We emphasize the differences from the proof of Thm.~\ref{thm:unicycle_CLF_BoPA}. To this end, let
\begin{align}
    U &=  U_1 + q^2 U_2, \;\; 
    U_1 = 2\tan^2 \frac{\delta}{2}, \;\; U_2 = 2\tan^2 \frac{\gamma}{2} \label{eq:U_AbounD}\\
    \Pi &= z^2\,, \quad z = 2\tan\frac{\delta}{2} +  2q \tan\frac{\gamma}{2}.
\end{align}
Note that $U(\gamma,\delta)$ is radially unbounded on $\mathcal{T}_3$. 
Then, the time derivative of \eqref{eq:U_AbounD} along the solutions of \eqref{eq:unicycle_polar_closed_loop-Gv-3-pass} is
\begin{equation}
\dot U  =  \frac{8 q^2 \sin \gamma}{(1+\cos\gamma)^2} \left[   k_3   \frac{(1+\cos\gamma)^2 \cos\gamma \sin \delta}{ (1+\cos \delta)^2 } - \tilde{\omega}
\right]\,. 
\end{equation}
We pick the control as in \eqref{eq-control-bounded-in-gamma-delta}, where we use the fact that $(1+\cos x )^2/4 = 1/(1+\tan^2 x/2)^2$ and obtain
\begin{equation}
\dot U  = -\frac{8 k_2  q^2 \sin^2\gamma}{(1+\cos\gamma)^2}   = - 2 k_2 q^2 U_2\,. \label{eq:U_dot_abound}
\end{equation}
With the feedback \eqref{eq-control-bounded-in-gamma-delta}, the LoS angle dynamics is 
\begin{equation}
\dot\gamma = - k_2 \sin\gamma -  k_3  \cos\gamma (1+\cos\gamma)^2  \frac{\tan \frac{\delta}{2}}{1 + \cos \delta} \,.
\end{equation}
Next, we consider the dynamics of the error  $z$, which are
\begin{equation}
\begin{aligned}[b]
\dot z = -&k_2 z + 2 k_2 \tan \frac{\delta}{2}\\
&- k_3 q \cos \gamma \frac{1 + \cos \gamma}{1 + \cos \delta} \left( 2\tan \frac{\delta}{2 } -  2q \tan \frac{\gamma}{2} \right).
\end{aligned}
\end{equation}
Then,
\begin{equation}
\begin{aligned}[b]
\frac{1}{2} \dot \Pi = - k_2z^2 &+  2k_2 z \tan \frac{\delta}{2} \\
&- 
k_3 q \cos\gamma \frac{1+\cos\gamma}{1 + \cos \delta} \left(U_1 - q^2U_2 \right) 
\end{aligned}
\end{equation}
Considering the fact that $2/(1 + \cos \delta) = 1 + \tan^2(\delta/2)$, we can use the same inequalities used to get \eqref{eq:Pi_dot_BoPa_temp} and
obtain
\begin{equation}
\hspace*{-0.3cm}
\begin{aligned}[b]
\frac{1}{2}\dot \Pi \le -  \frac{3}{4} k_2 \Pi +  &\frac{k_2  U_1}{1  + \cos \delta } \left( 1- \frac{k_3 }{k_2}  
q(1+\cos\gamma)\cos\gamma    \right)  \\ & + \frac{k_1 q U_2}{1 + \cos \delta}  (1+\cos \gamma) \cos \gamma 
 \,.
\end{aligned}
\end{equation}
 Taking into account \eqref{eq:bound2} in Lemma \ref{lemma:sinterm_upperbound} with $k_1 k_3 \ge  k_2^2$ along with $(1+\cos \gamma) \cos \gamma \tan^2 \frac{\gamma}{2}\le 8 \tan^2 \frac{\gamma}{2}$, 
yields
\begin{equation}
\dot{\Pi}_0 \le - \frac{3}{2}k_2 \Pi_0 + a   \left( 1 + U_1\right)^2 U_2 \,,
\end{equation}
with $a = \max \{k_1q, \sqrt{k_1 k_2}\}$. 
In the same manner as in the proof of Thm.~\ref{thm:unicycle_CLF_BoPA}, we can show that the time-derivative of \eqref{eq:CLF_ABounD_delta_gamma} along the solutions of \eqref{eq:unicycle_polar_closed_loop-Gv-3-pass}, \eqref{eq-control-bounded-in-gamma-delta} is such that
\begin{equation}
\dot V_{\delta \gamma} \le  - 16 a q^2 \tan^4\frac{\gamma}{2} -  \frac{3}{2}k_2 \left(\tan \frac{\delta}{2} + q \tan \frac{\gamma}{2}\right)^2
\label{eq:V_dot_AbounD} \,.
\end{equation}
Hence $V(\rho, \delta, \gamma) = \mathcal{V}(\rho^2, V_{\delta \gamma})$, $V(\rho, \delta, \gamma) = \mathcal{V}(V_{\delta \gamma},\rho^2)$ are strict CLFs for \eqref{eq:unicycle_polar_closed_loop-Gv-1} w.r.t. the input pair $(v/\rho, \omega)$  and
$\rho =  \delta = \gamma = 0$ is GAS on $\mathcal{S}_3$.
\end{proof}

\section{Integrator Forwarding Controllers}

\subsection{Global Forwarding (GloFo) Controller}


\begin{theorem}[GloFo CLFs] 
\label{thm:CLF_GloFo}
Consider the system \eqref{eq:unicycle_polar_closed_loop-Gv-1} in closed-loop with \eqref{eq-basic-v-control}, \eqref{eq-omega-general}, and \eqref{eq:GloFo} with arbitrary $k_1, k_2, k_3 > 0$.
The point $\rho = \delta = \gamma = 0$ is GAS on $\mathcal{S}$ in accordance with Def.~\ref{def-our-GAS}.
Furthermore, all the composite Lyapunov functions $V(\rho, \delta, \gamma) = \mathcal{V}(\rho^2, V_{\delta \gamma})$ and $V(\rho, \delta, \gamma) = \mathcal{V}(V_{\delta \gamma},\rho^2)$, for all functions $\mathcal{V}$ satisfying the conditions in Proposition \ref{prop:composite_Lyap_function}, and with $V_{\delta \gamma}$ defined as
\begin{equation}
V_{\delta \gamma}(\delta, \gamma)  = 
\left(\delta + \frac{k_1}{2k_2} {\rm Si}(2\gamma) \right)^2 + q^2 \gamma^2 
\,,  \quad q = \sqrt{\frac{k_1}{k_3}} \,, \label{eq:CLF_GloFo_gamma_delta}
\end{equation}
 are (globally) strict CLFs for \eqref{eq:unicycle_polar_closed_loop-Gv-1} w.r.t. the input pair $(v/\rho,\omega)$ in the sense of Def \ref{def-CLF}. 
\end{theorem}

\begin{proof}
Consider \eqref{eq:unicycle_polar_closed_loop-Gv-3-pass} and the 
forwarding transformation
\begin{equation}
\zeta =\delta+ \dfrac{k_1}{2k_2}\mbox{Si}(2\gamma), 
\end{equation} where ${\rm Si}(\cdot)$ is as in \eqref{eq:Si_function}.
Then, the open-loop system is
\begin{eqnarray}
\dot\zeta &=& \frac{k_1}{k_2}\mbox{sinc}(2\gamma)(k_2\gamma -\tilde\omega) \,,
\\
\dot\gamma &=& - \tilde\omega\,.
\end{eqnarray}
Choosing the control law \eqref{eq:GloFo}, we obtain
\begin{subequations}
\label{eq:closed_loop_glofo}
\begin{eqnarray}
\dot\zeta &=& -\frac{k_1k_3}{k_2}\mbox{sinc}^2(2
\gamma) \zeta \,,
\\
\dot\gamma &=& - k_2\gamma - k_3\mbox{sinc}(2\gamma)\zeta\,.
\end{eqnarray}
\end{subequations}
The time derivative of \eqref{eq:CLF_GloFo_gamma_delta}
along the solutions of \eqref{eq:closed_loop_glofo} is
\begin{equation}
\begin{aligned}[b]
\dot V_{\delta \gamma}  =- \frac{k_1k_2}{k_3}& \left[
\left(\frac{k_3}{k_2}\mbox{sinc}(2\gamma) \zeta\right)^2 +\gamma^2 \right.  \\ &\left. + \left(\frac{k_3}{k_2}\mbox{sinc}(2\gamma) \zeta +\gamma \right)^2 \label{eq:V_dot_forward1}
\right]\,,
\end{aligned}
\end{equation}
which is negative for all $(\rho,\delta,\gamma) \ne (0,0,0)$ in $\{\rho \ge 0\} \times \mathcal{T}$.
Analogous to the proof of Thm.~\ref{thm:unicycle_CLF_polynomial}, we conclude that $V(\rho, \delta, \gamma) = \mathcal{V}(\rho^2, V_{\delta \gamma})$ and $V(\rho, \delta, \gamma) = \mathcal{V}(V_{\delta \gamma},\rho^2)$ are strict CLFs for \eqref{eq:unicycle_polar_closed_loop-Gv-1} and  
$\rho =  \delta = \gamma = 0$ is GAS on $\mathcal{S}$.
\end{proof}

\subsection{Bounded-in-LoS by Forwarding (BoFo) Controller}

\begin{theorem}[BoFo CLFs] 
\label{thm:CLF_BoFo}
Consider the system \eqref{eq:unicycle_polar_closed_loop-Gv-1} in closed-loop with \eqref{eq-basic-v-control}, \eqref{eq-omega-general}, and \eqref{eq:BoFo} with arbitrary $k_1, k_2, k_3 > 0$.
The point $\rho = \delta = \gamma = 0$ is GAS on $\mathcal{S}$ in accordance with Def. \ref{def-our-GAS}.
Furthermore, all the composite Lyapunov functions $V(\rho, \delta, \gamma) = \mathcal{V}(\rho^2, V_{\delta \gamma})$ and $V(\rho, \delta, \gamma) = \mathcal{V}(V_{\delta \gamma},\rho^2)$, for all functions $\mathcal{V}$ satisfying the conditions in Proposition \ref{prop:composite_Lyap_function}, and with $V_{\delta \gamma}$ defined as
\begin{equation}
V_{\delta \gamma}(\delta, \gamma)  = 
\left(\delta + \frac{k_1}{k_2} \sin \gamma \right)^2+ 4q^2 \tan ^2\frac{\gamma}{2} 
\,,\label{eq:CLF_BoFo_gamma_delta}
\end{equation}
with $q = \sqrt{k_1/k_3}$, are (globally) strict CLFs for \eqref{eq:unicycle_polar_closed_loop-Gv-1} w.r.t. the input pair $(v/\rho,\omega)$ in the sense of Def. \ref{def-CLF}. 
\end{theorem}

\begin{proof}
Consider \eqref{eq:unicycle_polar_closed_loop-Gv-3-pass}
with the forwarding transformation
\begin{equation}
\zeta = \delta + \frac{k_1}{k_2}\int_0^{\tan\frac{\gamma}{2}} \frac{\sin(4\arctan\sigma)}{2\sigma}{\rm d}\sigma = \delta + \dfrac{k_1}{k_2} \sin \gamma \,,
\end{equation}
then, the transformed open-loop system reads as
\begin{eqnarray}
\dot\zeta &=& \frac{k_1}{k_2}\cos\gamma\left(k_2\sin\gamma -
\tilde\omega \right) \,,
\\
\dot\gamma &=& -  \tilde\omega\,.
\end{eqnarray}
Considering the control law \eqref{eq:BoFo}, 
the closed-loop system is
\begin{subequations}
    \label{eq:bofo_closed_loop}
\begin{eqnarray}
\dot\zeta &=& -\frac{k_1k_3}{k_2}\dfrac{\cos^2(\gamma)}{\left(1+\tan^2\frac{\gamma}{2}\right)^2} \zeta \,,
\\
\dot\gamma &=& - k_2\sin\gamma - k_3\dfrac{\cos(\gamma)}{\left(1+\tan^2\frac{\gamma}{2}\right)^2}\zeta\,.
\end{eqnarray}
\end{subequations}
The time derivative of \eqref{eq:CLF_BoFo_gamma_delta} along the solutions of \eqref{eq:bofo_closed_loop} is
\begin{align}
\dot V_{\delta \gamma}  = &- \frac{k_1k_2}{k_3}\Biggl[
\left(\frac{k_3}{k_2}\frac{\cos\gamma}{{1+\tan^2\frac{\gamma}{2}}} \zeta\right)^2+4\tan^2\frac{\gamma}{2}\nonumber\\
& + \left(\frac{k_3}{k_2}\frac{\cos\gamma}{{1+\tan^2\frac{\gamma}{2}}}\zeta +2\tan\frac{\gamma}{2} \right)^2
\Biggr]\,, \label{eq:V_dot_forward2}
\end{align}
which is negative for all $(\rho,\delta,\gamma) \ne (0,0,0)$ in $\{\rho \ge 0\} \times \mathcal{T}_1$.
Analogous to the proof of Thm.~\ref{thm:unicycle_CLF_polynomial}, we conclude that  $V(\rho, \delta, \gamma) = \mathcal{V}(\rho^2, V_{\delta \gamma})$, $V(\rho, \delta, \gamma) = \mathcal{V}(V_{\delta \gamma},\rho^2)$ are strict CLFs w.r.t. \eqref{eq:unicycle_polar_closed_loop-Gv-1} and  
$\rho =  \delta = \gamma = 0$ is GAS on $\mathcal{S}_1$.
\end{proof}

\section{Backstepping Controllers}

The idea of backstepping for the ``unconventional bounded integrator chain'' \eqref{eq:unicycle_polar_closed_loop-Gv-3-pass} is given by
\begin{equation}
\label{eq-bkst-step1}
\dot\delta = k_1 \frac{\sin(2\gamma)}{2} =k_1 \left[\frac{\sin(2(\gamma-z))}{2} +\psi(z,\gamma)z\right]\,,
\end{equation}
with the function \eqref{eq:psi_definition} where $z=\gamma - \alpha_0(\delta)$ is a backstepping transformation and $\alpha_0(\delta)$ is a ``stabilizing function,'' which renders the origin of the system $\dot\delta = k_1 \frac{\sin(2\alpha_0(\delta))}{2}$, namely, of  \eqref{eq:unicycle_polar_closed_loop-Gv-3n-pass} with $\gamma=\alpha_0(\delta)$, or \eqref{eq-bkst-step1} with $z=0$, GAS on $\mathbb{R}$ or on $(-\pi,\pi)$.

\subsection{A Global Backstepping (GloBa) Controller}

\begin{theorem}[GloBa CLFs] 
\label{thm:CLF_Globa}
Consider the system \eqref{eq:unicycle_polar_closed_loop-Gv-1} in closed-loop with \eqref{eq-basic-v-control}, \eqref{eq-omega-general}, and \eqref{eq-bkst-1} with arbitrary $k_1, k_2, k_3, k_4 > 0$.
The point $\rho = \delta = \gamma = 0$ is GAS on $\mathcal{S}$ in accordance with Def. \ref{def-our-GAS}.
Furthermore, all the composite Lyapunov functions $V(\rho, \delta, \gamma) = \mathcal{V}(\rho^2, V_{\delta \gamma})$ and $V(\rho, \delta, \gamma) = \mathcal{V}(V_{\delta \gamma},\rho^2)$, for all functions $\mathcal{V}$ satisfying the conditions in Proposition \ref{prop:composite_Lyap_function}, and with $V_{\delta \gamma}$ defined as
\begin{equation}
V_{\delta \gamma}(\delta, \gamma)  = 
\delta^2 + q^2 \left(\gamma+ \frac{1}{2}\arctan(2k_2 \delta)  \right)^2 
\,,  \label{eq:CLF_GloBa_gamma_delta}
\end{equation}
with $q = \sqrt{k_1/k_3}$ are (globally) strict CLFs for \eqref{eq:unicycle_polar_closed_loop-Gv-1} w.r.t. the input pair $(v/\rho,\omega)$ in the sense of Def.~\ref{def-CLF}. 
\end{theorem}

\begin{proof}
As before, our starting point for the development of \eqref{eq-bkst-1} and \eqref{eq:CLF_GloBa_gamma_delta} is \eqref{eq:unicycle_polar_closed_loop-Gv-3-pass}, and 
 the identity
\begin{equation}
\sin(\arctan(x)-y) = \frac{x\cos y - \sin y}{\sqrt{1+x^2}}\,.
\end{equation}
Considering \eqref{eq:unicycle_polar_closed_loop-Gv-3n-pass}, 
we note that the would-be feedback for the LoS angle (the `stabilizing function')
\begin{equation}
\gamma(\delta) = -\frac{1}{2} \arctan(2k_2\delta) 
\end{equation}
results in the globally asymptotically stable $\delta$-dynamics
\begin{equation}
\dot{\delta} = k_1\frac{\sin(2\gamma(\delta))}{2} 
= - \frac{k_1 k_2 \delta}{\sqrt{1+4 k_2 ^2\delta^2}}\,.
\end{equation}
Now, we introduce a backstepping change of the $\gamma$-variable
\begin{equation}\label{eq:backstepping_z1}
 z = \gamma + \frac{1}{2}\arctan(2 k_2 \delta).
\end{equation}
Taking into account \eqref{eq:psi_definition} and \eqref{eq:backstepping_z1}, 
we obtain
\begin{subequations}
\label{eq:backstep_temp}
\begin{eqnarray}
\dot{\delta} &=& - k_1 \left( \frac{k_2 \delta}{\sqrt{1+4k_2^2\delta^2}}
- z \psi(z,\gamma) \right) \,, \label{eq:delta_globa} \\
\dot{z} &=&  k_1 \frac{\sin(2\gamma)}{2} \frac{k_2}{1+4k_2^2\delta^2} - \tilde{\omega}\,.  \label{eq:z_globa}
\end{eqnarray}
\end{subequations}
The time derivative of \eqref{eq:CLF_GloBa_gamma_delta}
along the solutions of \eqref{eq:backstep_temp} is
\begin{equation}
\begin{aligned}[b]
\dot V_{\delta \gamma} = -&\frac{2 k_1 k_2  \delta^2}{\sqrt{1+4 k_2^2 \delta^2}} \\ &+   2 q^2 z \left[  k_3 \psi(z,\gamma) \delta   +  \frac{k_1 k_2 \sin (2\gamma)}{2 (1+4k_2^2\delta^2)} - \tilde{\omega}\right]\,.\label{eq-V1dot-bkst1}
\end{aligned}
\end{equation}
 Choosing $\tilde{\omega}$ as in \eqref{eq-bkst-1}, 
we get
\begin{equation}
\dot V_{\delta \gamma} = -\frac{2 k_1 k_2  \delta^2}{\sqrt{1+4k_2^2\delta^2}} -  2q^2 k_4 z^2\,, \label{eq:Backstepping_Lyapunov_function_dot}
\end{equation}
which is negative for all $(\rho,\delta,\gamma) \ne (0,0,0)$ in $\{\rho \ge 0\} \times \mathcal{T}$.
Analogous to the proof of Thm.~\ref{thm:unicycle_CLF_polynomial}, we conclude that $V(\rho, \delta, \gamma) = \mathcal{V}(\rho^2, V_{\delta \gamma})$ and $V(\rho, \delta, \gamma) = \mathcal{V}(V_{\delta \gamma},\rho^2)$ are strict CLFs for \eqref{eq:unicycle_polar_closed_loop-Gv-1} and  
$\rho =  \delta = \gamma = 0$ is GAS on $\mathcal{S}$.
\end{proof}  

\subsubsection*{Easier-to-interpret, more conservative backstepping controllers}

Next, we present an alternative backstepping design using the same CLF \eqref{eq:CLF_GloBa_gamma_delta}, offering clearer intuition than the less interpretable feedback \eqref{eq-bkst-1}.
\begin{proposition} \em
    Consider \eqref{eq:unicycle_polar_closed_loop-Gv-1} in closed-loop with \eqref{eq-basic-v-control} and 
    \begin{equation}
    \label{eq-omega-bkst-b}
 \omega = \left( k_4 + \frac{k_3}{2k_2}\frac{C^2}{N} + k_1|\psi|B\right) z \,,
\end{equation}
where 
\begin{align}
B(\delta) &= 1+\frac{k_2}{N(\delta)^2}, \quad N(\delta) = \sqrt{1+4k_2^2\delta^2} \,, \label{eq:N_delta_def} \\
C(\delta, \gamma) &=  \psi(\gamma, z) N(\delta) - \frac{k_1 k_2}{k_3} B(\delta)\,.
\end{align}
with $k_1, k_2, k_3, k_4 > 0$, and $\psi(z,\gamma)$ and $z$ defined in \eqref{eq:psi_definition} and \eqref{eq:backstepping_z1}, respectively.  Then, the point $\rho = \delta = \gamma = 0$ is GAS on $\mathcal{S}$ in accordance with Def.~\ref{def-our-GAS}.
\end{proposition}
\raggedbottom
\begin{proof}
    The outline of the proof is as follows. Observe that 
    \begin{equation}
        \sin(2\gamma) = 2 \left(\psi(z,\gamma) z- \frac{k_2\delta}{N(\delta)}\right) \,.
    \end{equation} Then, the time derivative of \eqref{eq:CLF_GloBa_gamma_delta} along the solutions of \eqref{eq:unicycle_polar_closed_loop-Gv-3} is such that
\begin{align}
&\dot V_{\delta \gamma} =   -\frac{2 k_1 k_2  \delta^2}{\sqrt{1+4 k_2^2 \delta^2}} \nonumber\\
&+   2 q^2 z\left[  k_3 \psi(z,\gamma) \delta + k_1  \frac{\sin (2\gamma)}{2} \left(1 + \frac{k_2}{1+4k_2^2\delta^2} \right)- \omega\right] \nonumber \\
&= -2 k_1 k_2 \frac{\delta^2}{N} +2 q ^2 k_3 \frac{\delta}{N}Cz + 2 q^2 k_1 \psi B z^2 - 2 q^2 z\omega \nonumber \\
&\leq  -  k_1 k_2 \frac{\delta^2}{N}+ 2 q ^2 z \left[\left( \frac{k_3}{2k_2}\frac{C^2}{N} + k_1|\psi|B\right) z - \omega
\right] \,, \label{eq-V1dot-bkst1a}
\end{align}
where the bound follows from applying Young's inequality to $\frac{\delta}{N}Cz$. 
Taking the feedback \eqref{eq-omega-bkst-b}
we arrive at
\begin{eqnarray}
\label{eq-V1dot-bkst1b}
\dot V_1 &\leq & - 2 k_1 k_2 \frac{\delta^2}{N} - 2 q^2 k_4 z^2
\,. 
\end{eqnarray}
The rest of the proof is similar to the proof of Thm.~\ref{thm:CLF_Globa}.
\end{proof}

The feedback \eqref{eq-omega-bkst-b}, featuring the backstepping-transformed LoS angle $z$ scaled by a nonlinear positive gain dependent on $\delta$ and $\gamma$, aggressively drives $\gamma$ toward $-\tfrac{1}{2}\arctan(2k_2\delta)$.
 Such an action by $\gamma$ is intuitive: 
 the vehicle is oriented toward the negative $x$-axis, which is advantageous since the target faces the positive direction; thus, motion with $\gamma = -\tfrac{1}{2}\arctan(2k_2\delta)$ guides the vehicle to the target “from behind,” eliminating the need to reverse. Looking quantitatively into the dependence of $\gamma$ on $\delta$ through $-\frac{1}{2}\arctan(2 k_2 \delta)$, we see that the closer the vehicle to the negative half of the $x$-axis the more directly the vehicle is made to point towards the target position, as it aims to arrive at it from behind. 
The feedback \eqref{eq-omega-bkst-b} is interpretable but complex due to its nonlinear dependence on $(C,N,\psi,B)$. Hence, we derive a simpler backstepping feedback. 
\begin{corollary}
\textit Consider
\eqref{eq:unicycle_polar_closed_loop-Gv-1} in closed-loop with \eqref{eq-basic-v-control} and 
\begin{equation}
\label{eq-omega-bkst-c}
\displaystyle \omega = \left[ k_4 + k_5 +\frac{k_3}{k_2}\left(1+4k_2^2\delta^2\right) \right] z \,,
\end{equation}
with $k_1, k_2, k_3, k_4 > 0$ and $k_5 = k_1 (1+k_2) \left[1+ \frac{k_1k_2(1+k_2)}{k_3}\right]$, and $\psi(z,\gamma)$, $z$ and $N(\delta)$ defined in \eqref{eq:psi_definition}, \eqref{eq:backstepping_z1} and \eqref{eq:N_delta_def}, respectively.  Then, the point $\rho = \delta = \gamma = 0$ is GAS on $\mathcal{S}$ in accordance with Def.~\ref{def-our-GAS}.
\end{corollary}
This corollary follows from \eqref{eq-V1dot-bkst1a}, by considering the conservative bounds
\begin{equation*}
|\psi| \leq  1,
\quad B \leq  1+ k_2,
\quad \frac{C^2}{N} \leq  2N^2 + 2 \frac{k_1^2 k_2^2 (1+k_2)^2}{k_3^2}\,,
\end{equation*}
which leads to
\begin{align}
\label{eq-V1dot-bkst1c}
\dot V_{\delta \gamma}
&\leq & -  k_1 k_2 \frac{\delta^2}{N}+ 2 q^2 z \left[\left( k_5 +\frac{k_3}{k_2}N^2\right) z - \omega
\right]
\,. 
\end{align}
Substituting \eqref{eq-omega-bkst-c} in \eqref{eq-V1dot-bkst1c},
we arrive at \eqref{eq-V1dot-bkst1b}. 
The feedback \eqref{eq-omega-bkst-c} is simpler and more intuitive—though more conservative and aggressive in steering—than \eqref{eq-omega-bkst-b}. Its nonlinear gain grows with $\delta^2$, meaning the feedback aligns the LoS angle more closely to $\gamma = -\tfrac{1}{2}\arctan(2 k_2 \delta)$ when the vehicle is farther from the negative $x$-axis, i.e., less directly behind the target.

\subsection{Backstepping to Avoid Running across Front Line (BAR-FLi) Controller}
\label{sec:LiBaC_BARFLi} 


\begin{theorem}[BAR-FLi CLFs] 
\label{thm:CLF_BARFLi}
Consider the system \eqref{eq:unicycle_polar_closed_loop-Gv-1} in closed-loop with \eqref{eq-basic-v-control}, \eqref{eq-omega-general}, and \eqref{eq-bkst-3}  with arbitrary $k_1, k_2, k_3, k_4 > 0$.
The point $\rho = \delta = \gamma = 0$ is GAS on $\mathcal{S}_2$ in accordance with Def.~\ref{def-our-GAS}.
Furthermore, all the composite Lyapunov functions $V(\rho, \delta, \gamma) = \mathcal{V}(\rho^2, V_{\delta \gamma})$ and $V(\rho, \delta, \gamma) = \mathcal{V}(V_{\delta \gamma},\rho^2)$, for all functions $\mathcal{V}$ satisfying the conditions in Proposition \ref{prop:composite_Lyap_function}, and with $V_{\delta \gamma}$ defined as
\begin{equation}
\hspace*{-0.34cm}
V_{\delta \gamma} (\delta, \gamma) = 4\tan^2\frac{\delta}{2} + q^2\left[  \gamma + \frac{1}{2}\arctan\left(4k_2\tan \frac{\delta}{2} \right)  \right]^2 \,, \label{eq:CLF_BARFLi_gamma_delta}
\end{equation}
with $q = \sqrt{k_1/k_3}$, are (globally) strict CLFs for \eqref{eq:unicycle_polar_closed_loop-Gv-1} w.r.t. the input pair $(v/\rho,\omega)$ in the sense of Def.~\ref{def-CLF}. 
\end{theorem}

\begin{proof}
Consider the backstepping transformation
\begin{align}\label{eq:z_BARFLi}
z =  \gamma + \frac{1}{2}\arctan\left(4k_2\tan\frac{\delta}{2}\right) 
\,.
\end{align} 
Then, by \eqref{eq:psi_definition} and \eqref{eq:z_BARFLi}, we get
\begin{subequations}
\label{eq:BarFli_open_loop}
\begin{eqnarray}
\dot{\delta} &=&  -k_1 \left( \frac{2k_2 \tan(\delta/2)}{N(\delta)} - z \psi(z,\gamma)\right)\,, \\
\dot{z} &=& \frac{k_1 \sin(2\gamma)}{2}\frac{k_2/\cos^2(\delta/2)}{N^2(\delta)} -\tilde{\omega} \,, \label{eq:zdot_BARFLi}
\end{eqnarray}
\end{subequations}
with $N(\delta) \coloneqq \sqrt{1 + 16 k_2^2 \tan^2(\delta/2)}$.
The time derivative of \eqref{eq:CLF_BARFLi_gamma_delta} along \eqref{eq:BarFli_open_loop}, with $\tilde{\omega}$ as in \eqref{eq-bkst-3}, is
\begin{align}\label{eq:V1_dot_cl_BARFLi}
\dot{V}_{\delta \gamma } &= \frac{-8k_1k_2(1 + \tan^2\frac{\delta}{2})\tan^2\frac{\delta}{2}}{N(\delta)} - 2k_4q^2z^2 
\,,
\end{align}
which is negative for all $(\rho,\delta,\gamma) \ne (0,0,0)$ in $\{\rho \ge 0\} \times \mathcal{T}_2$.
Analogous to Thm.~\ref{thm:CLF_Globa}, the result of Thm.~\ref{thm:CLF_BARFLi} follows.
\end{proof}

\subsubsection*{Linear-in-Angles Backstepping Controller (LiBaC)}

Observe that the backstepping transformation \eqref{eq:z_BARFLi} uses an arctangent of a tangent which, for a proper choice of $k_2$, is a nearly linear operation. Considering this, one can instead take a linear backstepping transformation of the form
\begin{align}\label{eq:z_LiBaC}
z = \gamma + \frac{1}{2}\delta\,.
\end{align}
We design a controller based on \eqref{eq:z_LiBaC}, which we refer to as Linear-in-Angles Backstepping Controller (LiBaC, pronounced `lie back').
\begin{proposition}[LiBaC] \em
Consider
\eqref{eq:unicycle_polar_closed_loop-Gv-1} in closed-loop with \eqref{eq-basic-v-control} and 
\begin{equation}
\label{eq:omega_LiBaC}
\begin{aligned}
\omega = k_3z + \dfrac{3k_1}{4}\sin(2\gamma) +  k_2 \frac{\tan (\delta/2)}{1 + \cos \delta} \psi(z,\gamma) \,, \end{aligned} 
\end{equation}
with $k_1, k_2, k_3> 0$, and  $\psi(\gamma, z)$ and $z$ defined in \eqref{eq:psi_definition} and \eqref{eq:z_LiBaC}, respectively. Then, the point  $\rho = \delta = \gamma = 0$ is GAS on $\mathcal{S}_3$ in accordance with Def.~\ref{def-our-GAS}.
\end{proposition}

\begin{proof}
The transformation \eqref{eq:z_LiBaC} yields $\dot\delta = -k_1/2 \sin\delta + k_1\psi(z,\gamma)z$, with $\delta \in (-\pi,\pi)$, and  $\dot{z} = \frac{3}{4}k_1\sin(2\gamma)- \omega$.
 Considering the Lyapunov function \eqref{eq:CLF_BARFLi_gamma_delta}, we obtain 
\begin{equation}\label{eq:Vdot_LiBaC}
\begin{aligned}[b]
\dot V_{\delta \gamma} = -k_1\tan^2\frac{\delta}{2} +2 q^2 z &\left[-\omega +\frac{3}{4}k_1\sin(2\gamma) \right.  \\ 
&\left. +k_2 \frac{\tan(\delta/2)}{1+\cos\delta} \psi(z,\gamma)
\right]\,.
\end{aligned}
\end{equation}
Then, choosing \eqref{eq:omega_LiBaC}
yields $\dot{V}_{\delta\gamma} = -k_1 \tan^2\frac{\delta}{2} -2k_1z^2$. The rest of the proof is analogous to that of Thm.~\ref{thm:CLF_BARFLi}.
\end{proof}

LiBaC achieves the same basin of attraction as BAR-FLi~\eqref{eq-bkst-3}, keeping $\abs{\delta} < \pi$, and avoids crossing the front-line of the target.


\section{ Local Eigenvalue Assignment}
\label{sec-linearization}



\begin{proposition}\em 
\label{prop:linearization_CLF_passivity_based}
The linearization for $\rho>0$ around the point $\rho =\delta = \gamma = 0$ of the closed-loop system~\eqref{eq:unicycle_polar_closed_loop-Gv-2},~\eqref{eq:unicycle_polar_closed_loop-Gv-3-pass} with the angular velocity \(\tilde{\omega}\) chosen from: Genova~\eqref{eq:angular_velocity_genova}, BolSa~\eqref{eq-control-bounded-in-gamma}, BoPa~\eqref{eq:Bopa_controller} or BAgAl~\eqref{eq-control-bounded-in-gamma-delta},
with gains $k_1, k_2, k_3 > 0$, is 
\begin{align}
\label{eq:unicycle_polar_closed_loop_l_passive}
\begin{bmatrix}
\dot{\rho}  \\
\dot{\delta} \\ 
\dot{\gamma} 
\end{bmatrix} = 
\begin{bmatrix}
-k_1 &0 &0 \\
0 &0 &k_1 \\
0 &-k_3 &-k_2
\end{bmatrix} \begin{bmatrix}
\rho \\ \delta \\ \gamma
\end{bmatrix}  \,.
\end{align}
Its eigenvalues 
$-p_1, -p_2,-p_3$ are such that
$k_1 = p_1 , 
k_2 = p_2 + p_3 ,
k_3  = {p_2 p_3}/{p_1}$.   
Under the condition for global stabilization $k_1 k_3 \ge k_2^2$ in Thms. \ref{thm:unicycle_CLF_polynomial},  \ref{thm:unicycle_CLF_BoLSA}, \ref{thm:unicycle_CLF_BoPA}, \ref{thm:unicycle_CLF_Bagal}, the linearization \eqref{eq:unicycle_polar_closed_loop_l_passive}
has one real pole $-p_1$, assigned with $k_1 = p_1$, 
and two conjugate-complex poles $-p_2, -\bar{p}_2$, assigned with $     k_2 = 2 \mathfrak{Re}\{p_2\}, 
k_3 = {\abs{p_2}^2}/{p_1}$.
\end{proposition}

\begin{proof} The relationship between the eigenvalues $-p_i$ and the gains $k_i$ is obtained using Vi\`{e}te's formulae for the second-order portion of the characteristic polynomial of \eqref{eq:unicycle_polar_closed_loop_l_passive}. Real eigenvalues require $k_2^2 = 4k_1k_3$ (damping ratio $\ge 1$), which is violated when $k_1 k_3 \ge k_2^2$ (damping ratio  $\le 1/2$). 
\end{proof}

While the eigenvalues can be assigned to arbitrary values {\em locally}, as with the linear-in-$(\rho,\delta,\gamma)$ design by Astolfi \cite[Prop. 1]{astolfi1999exponential}, when a strict GAS result is desired, the condition $k_1  k_2 \ge k_2^2$ results in two eigenvalues being complex. 


\begin{proposition}\em 
\label{prop:linearization_CLF_forwarding_based}
The linearization for $\rho>0$ around the point $\rho =\delta = \gamma = 0$ of the closed-loop system~\eqref{eq:unicycle_polar_closed_loop-Gv-2},~\eqref{eq:unicycle_polar_closed_loop-Gv-3-pass} with the angular velocity \(\tilde{\omega}\)  chosen from: GloFo~\eqref{eq:GloFo} or BoFo~\eqref{eq:BoFo},
with gains $k_1, k_2, k_3 > 0$, is 
\begin{align}
\label{eq:unicycle_polar_closed_loop_l_forwarding}
\begin{bmatrix}
\dot{\rho}  \\
\dot{\delta} \\ 
\dot{\gamma} 
\end{bmatrix} = 
\begin{bmatrix}
-k_1 &0 &0 \\
0 &0 &k_1 \\
0 &-k_3 &-k_2 - \dfrac{k_1k_3}{k_2}
\end{bmatrix} \begin{bmatrix}
\rho \\ \delta \\ \gamma
\end{bmatrix}  \,.
\end{align}
Its eigenvalues 
$-p_1, -p_2,-p_3$ are such that
$k_1 = p_1 , 
2k_2 = p_2 + p_3 \pm \abs{p_2 - p_3} ,
k_3  = {p_2 p_3}/{p_1}$.  
\end{proposition}

\begin{proof} The relationship between the eigenvalues $-p_i$ and the gains $k_i$ is obtained by using Vi\`{e}te's formulae for the second-order portion of the characteristic polynomial of \eqref{eq:unicycle_polar_closed_loop_l_forwarding}.
\end{proof}


\begin{proposition} \em
\label{prop:linearization_CLF_backstepping_based}
The linearization for $\rho > 0$ around the point $\rho =\delta = \gamma = 0$  of the closed-loop system~\eqref{eq:unicycle_polar_closed_loop-Gv-2} and~\eqref{eq:unicycle_polar_closed_loop-Gv-3-pass}, where the angular velocity  \(\tilde{\omega}\) is given by either  GloBa~\eqref{eq-bkst-1} or BAR-FLi~\eqref{eq-bkst-3},
with gains $k_1, k_2, k_3,k_4 > 0$, is 
\begin{equation}
\label{eq:unicycle_polar_closed_loop-l-back}
\hspace*{-0.25cm}
\begin{bmatrix}
\dot{\rho}  \\
\dot{\delta} \\ 
\dot{\gamma} 
\end{bmatrix} = 
\begin{bmatrix}
-k_1 &0 &0 \\
0 &0 &k_1 \\
0 &-(k_3 + k_2 k_4) &-(k_1 k_2 + k_4) 
\end{bmatrix} \begin{bmatrix}
\rho \\ \delta \\ \gamma
\end{bmatrix}\,.
\end{equation}
Its eigenvalues $-p_1$, $-p_2$, and $-p_3$ have negative real parts, with (w.l.o.g.) $p_3 \ge p_2$ if real. The four gains complete the eigenvalue assignment iff they 
satisfy 
$k_1 = p_1 ,\, 
k_1k_2 + k_4 = p_2 + p_3,\,
k_3 + k_2 k_4  = {p_2 p_3}/{p_1}$.
Infinitely many gains satisfy these conditions; for $\epsilon\in(0,\mathfrak{Re}\{p_2\})$ and $k_1=p_1$, one choice is
$k_2=(\mathfrak{Re}\{p_2\}-\epsilon)/p_1$, $k_4=\mathfrak{Re}\{p_3\}+\epsilon$, and
$k_3=[\epsilon^2+(p_3-p_2)\epsilon]/p_1$ when all eigenvalues are real, or $k_3=[\epsilon^2+(\mathfrak{Im}\{p_2\})^2]/p_1$ when $p_2=\bar p_3$ are complex conjugate.
\end{proposition}

\begin{proof}
By substitution of the gain values. 
\end{proof}


\begin{figure}[t]
\centering
\includegraphics[width=.8\linewidth]{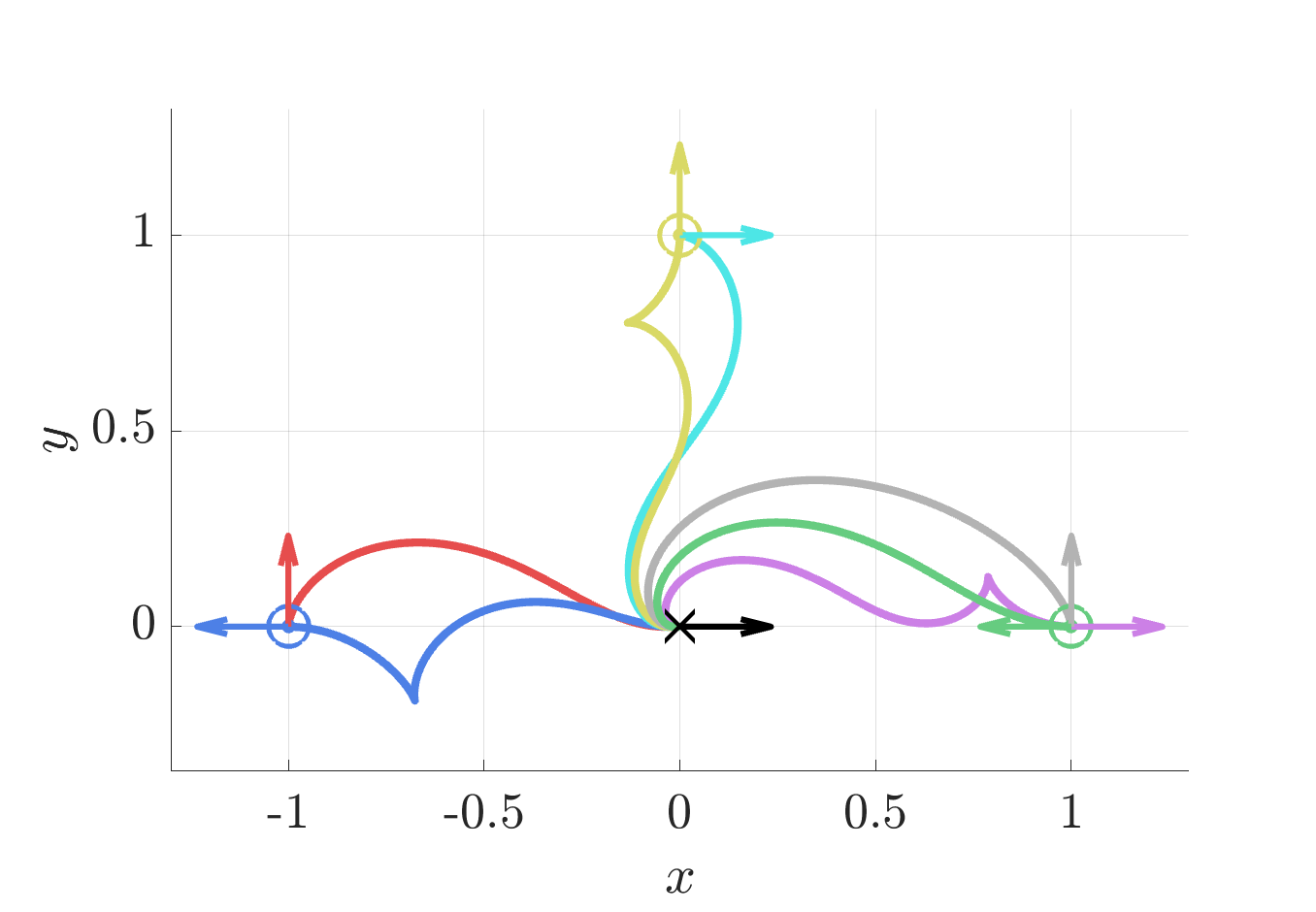}
\caption{Cartesian trajectories of GloBa~\eqref{eq-bkst-1} with unity gains.}
\label{fig:globa}
\end{figure}

\section{Comparing trajectories}
\label{sec:sims}
To validate the proposed control laws, we conducted numerical simulations comparing closed-loop trajectories under representative initial conditions, with the target position and heading angle shown in black. We first present the trajectories of GloBa~\eqref{eq-bkst-1} in Fig.~\ref{fig:globa}.
In Fig.~\ref{fig:barfli_abound}, we compare BAR-FLi~\eqref{eq-bkst-3} (blue) and BAgAl (cyan) against GloBa (red), where both BAR-FLi and BAgAl avoid crossing in front of the target, consistent with the intended effect of the barrier CLFs on the polar angle $\delta$.


\section{Conclusion}
We establish a modular design framework for  unicycle parking in polar coordinates, decoupling feedback design for distance and steering dynamics. By allowing bidirectional motion, for the steering subsystem we introduce a framework that incorporates all three of the most widespread nonlinear feedback methods---passivity, backstepping, and integrator forwarding---on both unconstrained and constrained state spaces. Our families of global strict barrier CLFs ensure global asymptotic stability with quantitative $\mathcal{KL}$-estimates of convergence rates and eigenvalues assignment at the target. Our modular method's composite Lyapunov function families pave the way for optimal and adaptive controllers in the paper's  Part II.

\begin{figure}[t]
\centering
\includegraphics[width=0.75\linewidth]{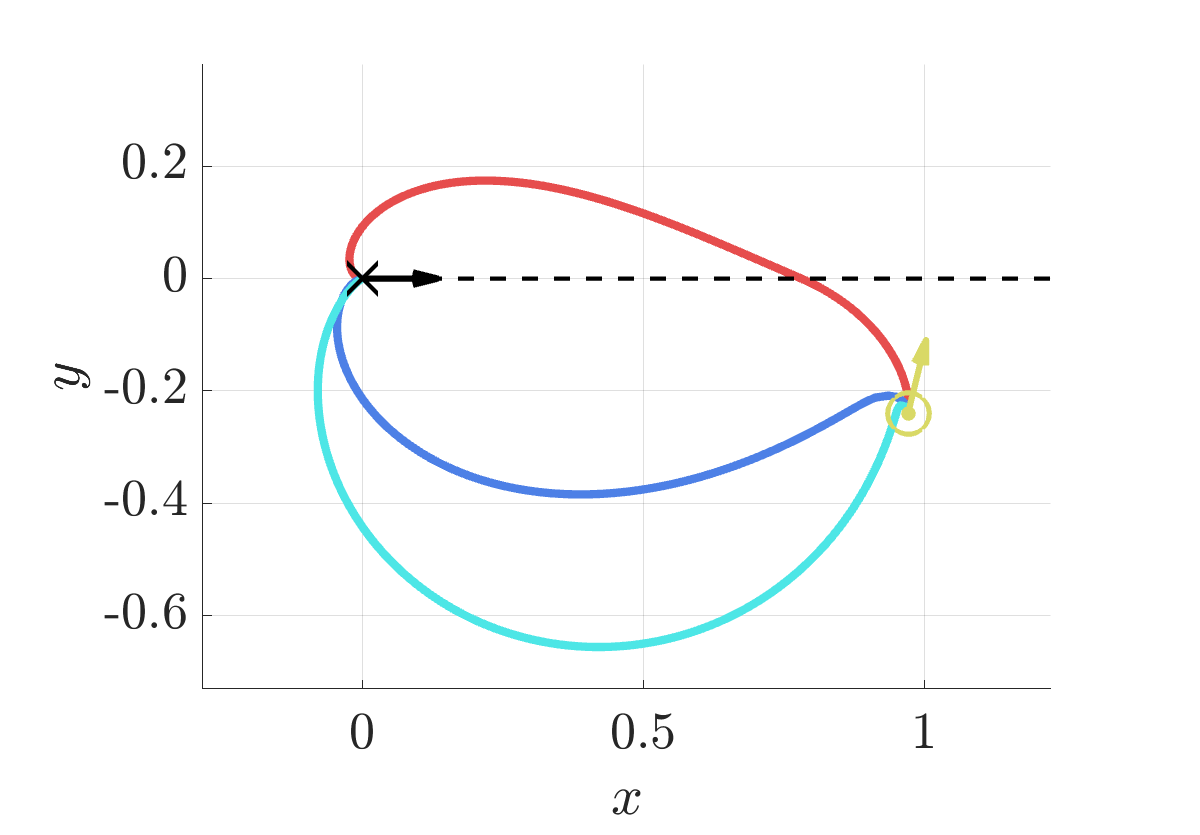}
\caption{Cartesian trajectories with BAR-FLi~\eqref{eq-bkst-3} (blue) and BAgAl~\eqref{eq-control-bounded-in-gamma-delta} (cyan) with gains $[k_1,k_2,k_3,k_4] = [1,1,0.1,1]$ compared to the GloBa~\eqref{eq-bkst-1} (red).}
\label{fig:barfli_abound}
\end{figure}

\vspace*{-0.5cm}

\section*{Appendix}

\renewcommand{\thedefinition}{A\arabic{definition}}
\renewcommand{\thelemma}{A\arabic{lemma}}

\begin{lemma}\label{lemma:sinterm_upperbound}
The following hold for all $k\geq 1$ and $x\in\mathbb{R}$:
\begin{eqnarray}
1-k \frac{\sin(2x)}{2x} &\le& k x^2 \label{eq:bound1}
\\ 
1-  
k \cos\gamma(1+\cos\gamma) &\leq& 2 (1 + k)  \tan^2\frac{\gamma}{2}\,. \label{eq:bound2}
\end{eqnarray} 
\end{lemma}

\begin{proof}
The derivation of \eqref{eq:bound1} is omitted as elementary and \eqref{eq:bound2} follows directly by substituting $\cos\gamma=\frac{1-\tan^2(\gamma/2)}{1+\tan^2(\gamma/2)}$.
\end{proof}

\section*{References}
\vspace{-0.6cm}
\bibliographystyle{IEEEtranS}
\bibliography{root}

\end{document}